\documentclass[12pt, draftclsnofoot, onecolumn]{IEEEtran}
\usepackage{epsfig,latexsym}
\usepackage{float}
\usepackage{indentfirst}
\usepackage{amsmath}
\usepackage{amssymb}
\usepackage{times}
\usepackage{subfigure}
\usepackage{psfrag}
\usepackage{hyperref}
\usepackage{cite}
\usepackage{lastpage}
\usepackage{fancyhdr}
\usepackage{color}
 \usepackage{amsthm}
\usepackage{bigints}
\newtheorem{theorem}{Theorem}
\newtheorem{Lemma}{Lemma}
\newtheorem{lemma}[Lemma]{$\mathbf{Lemma}$}

\newcounter{problem}
\newcounter{save@equation}
\newcounter{save@problem}
\makeatletter

\begin{document}
\title{ \LARGE{ Harvesting Devices' Heterogeneous Energy Profiles and QoS Requirements in  IoT: \\WPT-NOMA vs BAC-NOMA  }}

\author{ Zhiguo Ding, \IEEEmembership{Fellow, IEEE}    \thanks{ 
  
\vspace{-2em}

      Z. Ding
 is    with the School of
Electrical and Electronic Engineering, the University of Manchester, Manchester, UK (email: \href{mailto:zhiguo.ding@manchester.ac.uk}{zhiguo.ding@manchester.ac.uk}).

  }\vspace{-3em}}
 \maketitle
 
\begin{abstract}  
The next generation Internet of Things (IoT) exhibits a unique feature that IoT devices have different energy profiles and quality of service (QoS) requirements. 
 In this paper, two energy and spectrally efficient transmission strategies, namely   wireless power transfer assisted non-orthogonal multiple access (WPT-NOMA) and backscatter communication assisted  NOMA (BAC-NOMA), are proposed by utilizing this feature of IoT and employing spectrum and  energy cooperation among the  devices. Furthermore, for the proposed WPT-NOMA scheme, the application of hybrid successive interference cancelation (SIC) is also considered, and     analytical results  are developed to demonstrate that   WPT-NOMA  can avoid outage probability error floors and realize the full diversity gain. Unlike WPT-NOMA, BAC-NOMA suffers from an outage probability error floor, and the asymptotic behaviour of this   error floor   is analyzed in the paper by applying the extreme value theory. In addition, the effect of a unique feature of BAC-NOMA, i.e., employing  one device's signal as the carrier signal for   another device, is studied, and its   impact on the diversity gain  is revealed. Simulation results are also provided to compare the performance of the proposed   strategies  and verify the developed analytical results. 
\end{abstract} \vspace{-1em}

\section{Introduction} 
The next generation   Internet of Things (IoT)   is envisioned to support various important applications, including smart home, intelligent transportation,  wireless health-care, environment monitoring, etc \cite{ngiot}. The key   step to implement   the IoT is to ensure that a massive number of  IoT devices with heterogenous  energy profiles and quality of service (QoS) requirements  can be  connected  in a spectrally   efficient manner, which results in the two following challenges. From the spectral efficiency perspective, it is challenging to support massive connectivity, given the scarce bandwidth resources available for wireless communications. Non-orthogonal multiple access (NOMA) has been recognized as a spectrally efficient solution to support massive connectivity by encouraging spectrum sharing among wireless devices with different QoS requirements \cite{jsacnomaxmine,6666156,NOMAPIMRC}. For example, in conventional orthogonal multiple access (OMA), a delay-sensitive IoT device is allowed to solely  occupy a bandwidth resource block, which is not helpful to support massive connectivity and can also result in low spectral efficiency, particularly if this device has a small amount of data to send. By using NOMA, additional users, such as delay-tolerate  devices, can   be admitted to the channel. As a result, the overall spectral efficiency is improved,  and the use of advanced forms of NOMA can    ensure   that   massive connectivity  is supported while strictly guaranteeing  all devices'    QoS requirements  
\cite{8662677,8986647,8674774}. 

From the energy   perspective, the challenge is due to the fact that some IoT devices might be      equipped with continuous power supplies, but there are   many other devices  which  are battery powered and hence severely energy constrained. This challenge   motivates the use of two techniques, wireless power transfer (WPT) and backscatter communication (BackCom). The key idea of WPT  is to use radio frequency (RF) signals for energy transfer. In particular,  an energy-constrained IoT device can first carry out energy harvesting by using  the RF signals sent by a power station or another non-energy-constrained node in the wireless network, where the harvested energy can be used to power the transmission of the energy-constrained device \cite{Ruizhangbroadcast13,6951347,7081080,8758981}.   Similar to WPT,  BackCom is another low-power and low-complexity technique to connect   energy-constrained devices  \cite{7876867,7551180,xiangback}.   The key idea of BackCom is to ask an energy-constrained IoT device, termed a tag, to carry out  passive reflection and modulation of  a single-tone sinusoidal continuous wave sent by a BackCom reader. Instead of relying on the continuous wave sent by a reader, a variation of BackCom, termed symbiotic radio, was recently proposed to use the information-bearing signal sent by a non-energy-constrained  device   to power a batteryless device \cite{8907447,8399824}. 

In order to simultaneously address the aforementioned spectral and energy challenges, it is natural to consider the combination of NOMA with the two energy-cooperation transmission techniques in the next generation IoT, which will be the focus of this paper. Early examples of WPT assisted NOMA (WPT-NOMA) have considered the cooperative communication scenario, where relay transmission is powered by the energy harvested from the signals sent by a source \cite{yuanweijsac, 9007658,7917312}. In downlink scenarios, the use of WPT-NOMA can yield a significant improvement in the spectral and energy efficiency as demonstrated by   \cite{9104736,9091839,8891923}.  The application of WPT to uplink NOMA has been previously studied in   \cite{7582543}, where   users use the energy harvested from the signal sent by the base station to power  their uplink NOMA transmission. Compared to WPT-NOMA, the application of BackCom to NOMA received less attention. In \cite{8439079} and \cite{wongwcnc20}, NOMA was used to ensure that multiple backscatter devices can communicate with the same access point (a reader) simultaneously by modulating the  continuous wave sent by the access point. More recently, the application of NOMA to a special case of BackCom, symbiotic radio, has been considered in \cite{8962090,8636518}. 

The aim of this paper is to  consider a NOMA uplink   scenario, where a delay-sensitive non-energy-constrained  IoT device  and multiple delay-tolerant  energy-constrained  devices communicate with the same access point. In particular, following the semi-grant-free protocol proposed in \cite{8662677} and \cite{SGFx},  one of the delay-tolerant devices is granted access to the channel which would be solely occupied by the delay-sensitive device in OMA. Because some IoT devices are energy constrained, the use of the two energy-cooperative transmission strategies, WPT-NOMA and BAC-NOMA, is considered,   and their performance is compared. The contributions of the paper are listed as follows:

\begin{itemize}
\item A new WPT-NOMA scheme is proposed by applying hybrid successive interference cancellation (SIC), where the transmission of an energy-constrained device is powered by the energy harvested from the signals sent by the non-energy-constrained device. Recall that hybrid SIC  is to dynamically decide  the SIC decoding order by simultaneously using the devices' channel state information (CSI) and their QoS requirements    \cite{SGFx}. An intermediate benefit for using hybrid SIC for NOMA uplink is to avoid an outage probability error floor, which is not possible if a fixed SIC decoding order is used. In this paper, the outage performance of   WPT-NOMA with hybrid SIC is analyzed, and the obtained analytical results demonstrate that   outage probability error floors can be   avoided and the full diversity gain is still achievable, even though the transmission of those energy-constrained   devices are not powered by their own batteries. 

\item A general multi-user BAC-NOMA scheme is proposed, where an energy-constrained device reflects and modulates the signals sent by the non-energy-constrained device. Note that   the BAC-NOMA scheme considered in \cite{8636518} can be viewed as a special case of this general framework. In addition, the two key features of BAC-NOMA are analyzed in detail. Firstly, we focus on the outage probability error floor suffered by BAC-NOMA.  The key event which causes the error floor is analyzed, and the asymptotic behaviour of the probability of this event with respect to the number of the participating devices is studied by applying the extreme value theory (EVT) \cite{Arnoldbook,1710338}. Secondly, we focus on another feature of BAC-NOMA, i.e., modulating the energy-constrained  device's signal on the non-energy-constrained device's signal. This feature means that the relationship between the two devices' signals is multiplicative, instead of additive. Or in other words,  the non-energy-constrained  device's signal can     be viewed as a type of fast fading for   the energy-constrained device.   The analytical results developed in the paper show that this virtual fading is damaging to the reception reliability, and the diversity gain achieved by BAC-NOMA is capped by one, even if the event which causes the outage probability error floor can be ignored. 

\item The performance achieved by the two energy and spectrally efficient  transmission strategies is compared by using the provided analytical and simulation results. Our finding is that WPT-NOMA can offer a significant outage performance gain over BAC-NOMA, particularly at high signal-to-noise ratio (SNR) and with small target data rates, which is due to the fact that   hybrid SIC can be implemented in WPT-NOMA systems. However, WPT-NOMA suffers the two following drawbacks. One is that WPT-NOMA cannot support continuous transmission, which has a harmful impact on its ergodic data rate. The other is that WPT-NOMA is sensitive to how much time is allocated for energy harvesting and data transmission, respectively, where an inappropriate choice can lead to a significant performance loss, compared to BAC-NOMA.  
\end{itemize}

\section{System Model}\label{section system model}
Consider a  NOMA uplink scenario with one access point and $(M+1)$ IoT devices, denoted by  ${\rm U_m}$, $0\leq m\leq M$. For illustration purposes, assume that ${\rm U_0}$ is a non-energy-constrained delay-sensitive device, whereas    ${\rm U_m}$, $1\leq m\leq M$, are energy constrained and delay tolerant.   The channel from the access point to ${\rm U_i}$ is denoted by   $h_i$, $0\leq i \leq M$. The channel from ${\rm U_0}$ to ${\rm U_m}$ is denoted by $g_m$, $1\leq m\leq M$. 

Because  ${\rm U}_0$ is delay sensitive, it is allowed to solely occupy a bandwidth resource block in OMA, which is   spectrally inefficient for supporting massive connectivity.  Following the designs shown in \cite{8662677} and \cite{SGFx}, we consider that one of the delay-tolerant IoT   devices   is to be granted access to the resource block which would be solely occupied by ${\rm U}_0$ in OMA. 

{\it Assumption:} To facilitate performance analysis, we assume that the energy-constrained   devices are located in a small-size cluster, such that that the distances between ${\rm U}_0$ and ${\rm U}_m$, $m\geq 1$,  are same. A similar assumption is also made to  the distances between the access point and the devices. For example, the devices can be  sensors in a self-driving  vehicle or on an autonomous  robot. For smart home applications, the devices can be sensors for different functionalities  fixed in the same room.  Therefore, we assume that $g_m$, $1\leq m \leq M$, are  modelled as independent and identically  distributed (i.i.d.) Rayleigh fading, i.e., complex Gaussian distributed with zero mean and   variance $\lambda_g$, $g_m\sim CN(0,\lambda_g)$, where $\lambda_g\triangleq d_g^\phi$, $d_g$ denotes the distance between   ${\rm U}_0$ and ${\rm U}_m$, $m\geq 1$, and $\phi$ denotes the path loss exponent.  
Similarly, we also assume that $h_m\sim CN(0,\lambda_h)$ and $h_0\sim CN(0,\lambda_0)$, where $\lambda_h\triangleq d_h^\phi$,   $\lambda_0\triangleq d_0^\phi$,  $d_h$ denotes the distance between   the access point and  ${\rm U}_m$, $m\geq 1$,   and $d_0$ denotes the distance between ${\rm U}_0$  and the access point.

\subsection{WPT Assisted NOMA}
Without loss of generality, assume that ${\rm U}_m$ is granted access, where the details for the scheduling strategy will be provided at the end of this subsection. Suppose that the energy-constrained   devices can support WPT, and  
  time-switching WPT is used for its simplicity, which consists of two phases \cite{Zhouzhang13}. During the first  $\alpha T$ seconds, ${\rm U}_m$ performs energy harvesting by using ${\rm U}_0$'s signal, denoted by $s_0$,  and then uses the harvested energy for its transmit power to send its signal $s_m$ to the access point, where $\alpha$ denotes the time-switching parameter, $0\leq \alpha\leq 1$ and $T$ denotes the block period. Therefore, the amount of energy harvested at ${\rm U}_m$ is $\eta P|g_m|^2\alpha T$, where $P$ denotes ${\rm U}_0$'s transmit power,  $\eta$ denotes the energy harvesting efficiency coefficient. This means that the observation at the access point is given by
\begin{align}
y_{{\rm AP}} =\sqrt{P}h_0s_0+ \sqrt{ \frac{\eta P|g_m|^2\alpha }{1-\alpha}} h_m s_m +n_{{\rm AP}},
\end{align}
where $n_{{\rm AP}}$ denotes the noise. 

For the proposed WPT-NOMA scheme, hybrid SIC is applied \cite{hsic01,hsic02}. In particular, if $s_m$ is decoded first, ${\rm U}_m$'s maximal data rate without causing the failure of SIC (or degrading ${\rm U}_0$'s performance)  is given by
\begin{align} \label{eqwp1}
{\rm R}^{WP,1}_m =(1-\alpha)\log\left( 1+  \frac{\eta P\bar{\alpha}|g_m|^2 |h_m|^2}{P|h_0|^2+1} \right) , 
\end{align}
where $\bar{\alpha}=\frac{\alpha}{1-\alpha}$ and the noise power is assumed to be normalized. If  ${\rm U}_0$'s signal is decoded first, ${\rm U}_0$'s achievable  data rate is given by
\begin{align}\label{eqwp2}
{\rm R}^{WP,2}_{0,m} =(1-\alpha)\log\left(1+ \frac{P|h_0|^2}{\eta P\bar{\alpha}|g_m|^2 |h_m|^2+1}\right). 
\end{align}
 Denote ${\rm U }_0$'s target data rate by $R_0$. If  ${\rm R}^{WP,2}_{0,m}\geq R_0$, $s_0$ can be successfully decoded and removed, which means that $s_m$ can be decoded correctly with the following data rate:
\begin{align} \label{eqwp3}
{\rm R}^{WP,2}_m = (1-\alpha)\log\left(1+  \eta P\bar{\alpha}  |g_m|^2|h_m|^2 \right). 
\end{align}

\subsubsection*{ Device Scheduling for WPT-NOMA} The aim of device scheduling is to ensure that the delay-tolerant device which yields the largest data rate can be selected, under the condition that ${\rm U}_0$'s QoS requirements are strictly guaranteed. 
Note that 
${\rm R}^{WP,2}_{0,m} \geq R_0$ is equivalent to the following inequality: 
\begin{align}
   \gamma_m\leq   \frac{|h_0|^2}{\bar{\epsilon}_0\eta \bar{\alpha}} -\frac{1}{\eta P\bar{\alpha}} ,
\end{align}
where $\bar{\epsilon}_0 = 2^{\frac{R_0}{1-\alpha}}-1$. Furthermore,   define $\bar{\epsilon}_s=2^{\frac{R_s}{1-\alpha}}-1$ and  $\tau(h_0) = \max\left\{0, \frac{|h_0|^2}{\bar{\epsilon}_0\eta \bar{\alpha}} -\frac{1}{\eta P\bar{\alpha}}  \right\}$, where it is assumed that  ${\rm U}_m$, $1\leq m\leq M$, have the same target data rate, denoted by $R_s$. 
The delay-tolerant IoT devices can be divided into the two  groups, denoted by $\mathcal{S}_1$ and $\mathcal{S}_2$, respectively, as defined in the following: 
\begin{itemize}
\item $\mathcal{S}_1$ contains the devices whose channel gains satisfy  $\gamma_m>\tau(h_0)$. If one device from $\mathcal{S}_1$  is scheduled,  its signal has to be  decoded at the first stage of SIC, which yields the data rate ${\rm R}^{WP,1}_m$. 

\item $\mathcal{S}_2$ contains the devices whose channel gains satisfy  $\gamma_m\leq \tau(h_0)$. If one devices from $\mathcal{S}_2$ is scheduled,  its signal can be   decoded   either at   the first stage of SIC (which yields the data rate ${\rm R}^{WP,1}_m$) or at  the second stage of SIC (which yields the data rate ${\rm R}^{WP,2}_m$). Since $ {\rm R}^{WP,1}_m\leq {\rm R}^{WP,2}_m$ always holds, ${\rm U}_m$ always prefers its signal to be decoded at the second stage of SIC.

\end{itemize}

The access point selects the delay-tolerant device which yields the largest data rate, i.e.,
\begin{align} \label{swipt select}
m^* = \arg \max \left\{\max\left\{{\rm R}^{WP,1}_m, m\in\mathcal{S}_1\right\} ,\max\left\{{\rm R}^{WP,2}_m, m\in\mathcal{S}_2\right\} \right\}.
\end{align}

{\it Remark 1:} As can be observed from \eqref{eqwp1}, \eqref{eqwp2}, and \eqref{eqwp3}, the use of time-switching reduces the time duration for data transmission, since the first $\alpha T$ seconds are  used for energy harvesting. This feature of WPT-NOMA can lead to a potential performance loss compared BAC-NOMA which can support continuous data transmission.

\subsection{BackCom-Assisted NOMA}
Again  assume that ${\rm U}_m$ is granted access, where the details for the BAC-NOMA scheduling strategy will be provided later. 
Suppose that the energy-constrained  devices are capable to carry out backscatter communications. Therefore, the access point receives the following signal:\footnote{We assume that the symbol periods of different devices are   same, where the design of BAC-NOMA for the case with devices using different symbol periods is beyond  the scope of this paper. }
\begin{align}\label{eq1}
y_{{\rm AP}} = \sqrt{P}h_0s_0 +  \sqrt{P}\beta g_m h_ms_0s_m +n_{{\rm AP}},
\end{align}
where  $\beta$ denotes the BackCom power reflection coefficient. Unlike WPT-NOMA,  in BAC-NOMA, there is only one choice for the SIC decoding order,  which is to decode ${\rm U}_0$'s signal first. The reason for this is that   ${\rm U}_0$'s signal can be viewed as a fading channel for ${\rm U}_m$'s signal. In order to implement coherent detection,  ${\rm U}_0$'s signal, i.e., the virtual fading channel, needs to be decoded first.   
Therefore, in BAC-NOMA,  ${\rm U}_0$'s achievable data rate is given by
\begin{align}
{\rm R}_{0,m}^{BAC} = \log\left(1+\frac{P|h_0|^2}{P\beta^2 |g_m|^2 |h_m|^2+1}\right).
\end{align}

Assuming that  ${\rm U}_0$'s signal can be correctly decoded, i.e., $ {\rm R}^{BAC}_{0,m} \geq R_0$,     ${\rm U}_0$'s signal can be removed, which leads to the following system model:  
\begin{align}\label{eq2}
y_{{\rm AP}} - \sqrt{P}h_0s_0  =   \sqrt{P}\beta g_m h_ms_0s_m +n_{{\rm AP}}.
\end{align}
Therefore, an achievable data rate  for decoding $s_m$ is given by
\begin{align}\label{eq3}
{\rm R}^{BAC}_m =   \log\left(1+ {P}\beta^2 |g_m |^2 |h_m|^2|s_0|^2 \right), 
\end{align}
where ${\rm U}_0$'s signal, $s_0$, is viewed as a fast fading channel gain for $s_m$. Similar to  \cite{8907447,8636518}, it is assumed that $s_m \sim CN(0,1)$, i.e., the probability density function (pdf) of this virtual fading channel, $|s_0|^2$, is     $f_{|s_0|^2}(x)=e^{-x}$.

\subsubsection*{ Device Scheduling for BAC-NOMA}  In OMA, ${\rm U}_0$ is allowed to solely occupy the channel, whereas the use of NOMA   ensures that the backscatter devices can also be granted access. In order to guarantee    ${\rm U}_0$'s QoS requirements,  device ${\rm U}_m$ can be granted access only if   ${\rm R}_{0,m}^{BAC}\geq R_0$ which can be rewritten as follows: 
\begin{align}
|g_m|^2 |h_m|^2 \leq  \beta^{-2} \epsilon_0^{-1}|h_0|^2- \beta^{-2}P^{-1}  
\end{align}  
where $\epsilon_0=2^{R_0}-1$.  

On the other hand, it is ideal to admit  the   device which can maximize  the data rate $ {\rm R}^{BAC}_m $. 
Therefore, the device scheduling criterion is given by
\begin{align}\label{bd selection}
m^* = \arg \max\left\{ {\rm R}^{BAC}_m , m\in \mathcal{S}_0\right\},
\end{align}
where $\mathcal{S}_0=\left\{ m:  {\rm R}^{BAC}_{0,m} \geq R_0,1\leq m \leq M \right\}$.

{\it Remark 2:}  Unlike WPT-NOMA, BAC-NOMA can support one SIC decoding order only, which is the reason why it suffers an outage probability error floor, as shown in the next section.   Another feature of BAC-NOMA is that   $s_0$   is treated as a virtual fading channel, which means  $s_m$ suffers additional fading attenuation. The impact of this virtual fading channel on the reception reliability of $s_m$ will be investigated in the following section.  

{\it Remark 3:} We note that the two proposed device scheduling strategies  can be carried out  in a distributed manner. Take BAC-NOMA as an  example. Each backscatter device decides   to participate in contention, if $  {\rm R}_{0,m}^{BAC} >R_0, m\in \mathcal{S}$, otherwise it switches  to the match state. Each device calculates its backoff time inversely proportionally to its achievable data rate $ {\rm R}^{BAC}_m$, which ensures that   $ {\rm U}_{m^*}$ can be granted access in a distributed manner. 

\section{Performance Analysis for WPT-NOMA}
Since the implementation of WPT-NOMA is transparent to ${\rm U}_0$, we only focus on the performance of the admitted delay-tolerant   energy-constrained   device. Denote the  effective channel gains of the devices  by  $ \gamma_m=|g_m |^2 |h_m|^2$.  In order to simplify notations, without of loss of generality,   assume that the delay-tolerant devices  are ordered according to their effective channel gain as follows:
\begin{align} \label{channel order}
\gamma_1\leq \cdots\leq \gamma_M.
\end{align}
With this channel ordering, the impact of device scheduling on the NOMA transmission can be  shown  explicitly. Particularly, denote   $\bar{E}_m$ by the event that the size of $\mathcal{S}_2$ is $m$, i.e., $\bar{E}_m$ can be expressed as follows:
\begin{align}
\bar{E}_m = \left\{ \gamma_{m}<\tau(h_0), \gamma_{m+1}>\tau(h_0) \right\},
\end{align}
for $1\leq m\leq M-1$, where   $\bar{E}_0 = \left\{ \gamma_{1}>\tau(h_0)  \right\}$ and $\bar{E}_M= \left\{  \gamma_{M}<\tau(h_0) \right\}$. 

The outage probability achieved by WPT-NOMA can be expressed as follows: 
\begin{align}\nonumber
{\rm P}^{WP} =&\sum^{M}_{m=1} \underset{T_m}{\underbrace{{\rm P}\left(  \max\left\{ {\rm R}_{m} ^{WP,2},  {\rm R}_{M} ^{WP,1}\right\}<R_s , |\mathcal{S}_2|=m\right) }}\\  &+\underset{T_0}{\underbrace{{\rm P}\left(     {\rm R}_{M} ^{WP,1} <R_s , |\mathcal{S}_2|=0\right)  }}  . \label{wp outage}
\end{align}

We note that the performance analysis requires the pdf and cumulative distribution function (CDF) of the ordered channel gain $\gamma_m$, which can be found by using the density functions of the unordered channel gain. In particular, the pdf of the unordered effective channel gain   is given by \cite{8636518}
\begin{align}\label{pdf gamma} 
f_{\gamma}(x) =   2 \lambda_h\lambda_g K_0\left(2\sqrt{  \lambda_h\lambda_g x}\right),
\end{align}
where $K_i(\cdot)$ denotes the $i^{\rm th}$-order modified Bessel function of the second kind. 
The  CDF of the unordered channel gain, denoted by  $F_{\gamma}(x) $,  can be obtained as follows:
\begin{align}\nonumber 
F_{\gamma}(x)  =& \int^{x}_{0}  2\lambda_h\lambda_gK_0\left(2\sqrt{  \lambda_h\lambda_gy}\right) dy= \frac{4}{\lambda_h\lambda_g} x \int^{1}_{0} K_0\left(2t\sqrt{x}\sqrt{ \frac{1}{\lambda_h\lambda_g}}\right) tdt
\\\label{cdf gamma} 
=&1 - 2 \sqrt{ \lambda_h\lambda_g x}K_1\left(2 \sqrt{ \lambda_h\lambda_g x}\right)  ,
\end{align}
where   \cite[(6.561.8)]{GRADSHTEYN} is used. As can be observed from \eqref{pdf gamma} and \eqref{cdf gamma}, the density functions of the unordered channel gains contain Bessel functions, which makes it difficult to obtain an exact expression for the outage probability achieved by WPT-NOMA. However, the diversity gain achieved by WPT-NOMA can be obtained, as shown in the following theorem.  

\begin{theorem}\label{theorem1}
For the considered NOMA uplink scenario,    WPT-NOMA can realize a diversity gain of $M$, if $\bar{\epsilon}_0 \bar{\epsilon}_s<1$.
\end{theorem}
\begin{proof}
See Appendix \ref{proof1}. 
\end{proof}

{\it Remark 3:} Theorem \ref{theorem1} shows that the diversity gain achieved by WPT-NOMA is not zero, which implies that WPT-NOMA does not suffer any outage probability error floors, a feature   not achievable to  BAC-NOMA, as shown in the next section. Therefore, WPT-NOMA is a more robust transmission solution, compared to BAC-NOMA, particularly at high SNR. 

{\it Remark 4:} Note that $M$ is the maximal multi-user diversity gain achievable to the considered NOMA uplink scenario, since there are $M$ delay-tolerant devices competing for the access.  Theorem \ref{theorem1} shows that the maximal diversity gain can be realized by WPT-NOMA, even though battery-less transmission is used. Therefore, WPT-NOMA is particularly attractive for energy-constrained  IoT devices which have strict requirements for reception reliability.  


{\it Remark 5:} We note that the conclusion that there is no outage probability error floor  also holds for the special case $M=1$, i.e., there is a single delay-tolerant device and device scheduling is not carried out.  This implies that the outage probability error floor is avoided  due to the use of hybrid SIC, instead of   device scheduling

\section{Performance Analysis for BAC-NOMA}
Again because   the implementation of NOMA is transparent to ${\rm U}_0$, we only focus on the performance of the admitted delay-tolerant  device.  The outage probability of interest is expressed as follows:
\begin{align}
{\rm P}^{BAC} = {\rm P}\left(  {\rm R}_{m^*} ^{BAC}<R_s , |\mathcal{S}_0|\neq 0\right) + {\rm P}\left(  |\mathcal{S}_0|= 0\right) ,
\end{align}
where $|\mathcal{S}|$ denotes the size of set $\mathcal{S}$.

Assume that the devices' channel gains are ordered as in \eqref{channel order}. Denote   $E_m$ by the event that the size of $\mathcal{S}_0$ is $m$, i.e., $E_m$ can be expressed as follows:
\begin{align}
E_m = \left\{ \gamma_{m}<\theta(h_0), \gamma_{m+1}>\theta(h_0) \right\},
\end{align}
for $1\leq m\leq M-1$, where $\theta(h_0)=\beta^{-2} \epsilon_0^{-1}|h_0|^2- \beta^{-2}P^{-1} $. We note that $E_0 = \left\{ \gamma_{1}>\theta(h_0)  \right\}$ and $E_M= \left\{  \gamma_{M}<\theta(h_0) \right\}$. 

The use of \eqref{bd selection} and \eqref{channel order} means that $ {\rm U}_{m}$ will be granted access, for the   event $E_m$. Therefore, the outage probability can be further written as follows:
\begin{align}\label{outage bdd}
{\rm P}^{BAC} =& \sum^{M}_{m=1}\underset{Q_m}{\underbrace{{\rm P}\left( {\rm R}^{BAC}_{m} <R_s , E_m\right) }} + {\rm P}\left( E_0\right) . 
\end{align}

We note that ${\rm P}^{BAC}$ is more challenging to analyze, compared to ${\rm P}^{WP}$, because there are more random variables involved. In the following, we focused on two key features of WPT-NOMA.

\subsection{Outage Probability Error Floor}
In this subsection, we will show that BAC-NOMA suffers from an outage probability error floor. The existence of the   error floor can be sufficiently proved by focusing on a lower bound on the outage probability as shown in the following:
\begin{align}\label{floor1}
{\rm P}^{BAC} \geq  {\rm P}\left( E_0\right) . 
\end{align}

The simulation results provided in Section \ref{section simulation} show that      $E_0$ is indeed the most damaging event  at high SNR, compared to the terms $Q_m$, $1\leq m \leq M$.   $ {\rm P}\left( E_0\right) $ can be expressed as follows: 
\begin{align}
 {\rm P}\left( E_0\right)  =& {\rm P}\left(\gamma_1>\beta^{-2} \epsilon_0^{-1}|h_0|^2- \beta^{-2}P^{-1}\right)
\\\nonumber 
=& {\rm P}\left(\beta^2 \epsilon_0\gamma_1+  \epsilon_0 P^{-1}> |h_0|^2>\epsilon_0P^{-1} \right)  \\\nonumber &+{\rm P}\left(  |h_0|^2<\epsilon_0P^{-1} \right) .
\end{align}
Denote $f_{\gamma_1}(x) \triangleq  M f_{\gamma}(x)  \left(1-F_{\gamma}(x)\right)^{M-1}$ by the marginal pdf of the smallest order statistics, and hence $ {\rm P}_{E_0}  $ can be expressed as follows:
\begin{align}\nonumber
 {\rm P}\left( E_0\right)  
=& \int^{\infty}_{0}\left(e^{-\lambda_0\epsilon_0P^{-1}  }- e^{ -\lambda_0(  \beta^2 \epsilon_0x+  \epsilon_0 P^{-1})} \right)f_{\gamma_1}(x)dx \\\nonumber& +1 - e^{-\lambda_0\epsilon_0P^{-1} }
\\\nonumber  
=& 1 -Me^{ -\lambda_0   \epsilon_0 P^{-1}} \int^{\infty}_{0} e^{ -\lambda_0 \beta^2 \epsilon_0x } f_{\gamma}(x)\\ &\times \left(1-F_{\gamma}(x)\right)^{M-1}dx  .
\end{align}

At high SNR, i.e., $P\rightarrow \infty$,  $ {\rm P}\left( E_0\right) $ can be approximated as follows:
\begin{align}\nonumber
 {\rm P}\left( E_0\right) 
\approx & 1 -M  \int^{\infty}_{0} e^{ -\lambda_0 \beta^2 \epsilon_0x } f_{\gamma}(x)\left(1-F_{\gamma}(x)\right)^{M-1}dx  
\\ \label{floor2} 
\approx &   \lambda_0 \beta^2 \epsilon_0\int^{\infty}_{0} e^{ -\lambda_0 \beta^2 \epsilon_0x }  \left(1-F_{\gamma}(x)\right)^{M} dx,
\end{align}
which is constant and not a function of $P$. Combining \eqref{floor1} with \eqref{floor2}, it is sufficient to conclude  that BAC-NOMA transmission suffers an outage probability error floor. 

{\it Remark 6:} This finding is  consistent to the conclusions made in \cite{8636518}. The reason for the existence of this error floor is due to the fact that only one SIC decoding order can be used by BAC-NOMA. Compared to BAC-NOMA, WPT-NOMA can avoid this error floor and hence  outperform BAC-NOMA  at high SNR. 

{\it Remark 7:} Theorem \ref{theorem1}  indicates that WPT-NOMA can utilize the multi-user diversity, and hence a nature question  is whether BAC-NOMA can also use the multi-user diversity, i.e., whether it is beneficial to invite more delay-tolerant  devices to participate in transmission in BAC-NOMA.   By applying the EVT, the following lemma can be obtained for this purpose. 

\begin{lemma}\label{lemma1}
The error floor caused by $ {\rm P}\left( E_0\right) $ can be reduced to zero by increasing the number of participating delay-tolerant  devices $M$ and the transmit power $P$. 
\end{lemma}
\begin{proof}
See Appendix \ref{lemma1proof}.
\end{proof}

\subsection{Impact of $s_0$ on Reception Reliability}
Recall that $s_0$ is treated as a type of fast fading when the signal from the delay-tolerant device is decoded. In this section, we will show that this  fast fading has a harmful impact on the outage probability. To obtain  an insightful conclusion, we consider an ideal situation, in which $E_0$ does not happen. We will show that even in such an ideal situation, the full multi-user diversity gain cannot be realized. Recall     the term $Q_m$,    $1\leq m \leq M-1$, shown in \eqref{outage bdd} can be evaluated as follows: 
 \begin{align}\label{eq89}
Q_m =&  {\rm P}\left( {\rm R}^{BAC}_{m} <R_s ,  \gamma_{m}<\theta(h_0), \gamma_{m+1}>\theta(h_0)\right) \\\nonumber
=&  {\rm P}\left( \gamma_m <\min\{ \epsilon_s {P}^{-1}\beta^{-2}  |s_0|^{-2},  \theta(h_0)\}, \gamma_{m+1}>\theta(h_0)\right) .
\end{align}
Define $a_{s_0,h_0} = \min\{ \epsilon_s {P}^{-1}\beta^{-2}  |s_0|^{-2},  \theta(h_0)\}$. By applying order statistics, the joint pdf of $\gamma_m $ and $\gamma_{m+1}$ is given by \cite{Arnoldbook}
\begin{align}
f_{\gamma_m,\gamma_{m+1}}(x,y) =&\mu_0 f_{\gamma}(x) f_{\gamma}(y)\left(F_{\gamma}(x) \right)^{m-1}\\\nonumber &\times 
\left(1- F_{\gamma}(y) \right)^{M-m-1},
\end{align}
for $x<y$, 
where $\mu_0 =  \frac{M!}{(m-1)! (M-m-1)!} $.  
Therefore, $Q_m$ can be expressed as follows: 
 \begin{align}
Q_m   
=& \bar{\mu}_0\mathcal{E}_{h_0, s_0}\left\{  \int^{\infty}_{\theta(h_0)}f_{\gamma}(y)\left(1- F_{\gamma}(y) \right)^{M-m-1}dy\right.\\ &\times\left. \int^{a_{s_0,h_0}}_{0}
 f_{\gamma}(x) \left(F_{\gamma}(x) \right)^{m-1}
dx
\right\}
\\\nonumber
=&\bar{\mu}_0\mathcal{E}_{h_0, s_0}\left\{  \left(1- F_{\gamma}(\theta(h_0)) \right)^{M-m}  
  \left(F_{\gamma}(a_{s_0,h_0}) \right)^{m}
\right\} ,
\end{align}
where $\bar{\mu}_0=\mu_0\mathcal{E}_{h_0, s_0}$. Because   the density functions of $\gamma_m$ contain Bessel functions, a closed-form expression for $Q_m$ is difficult to obtain, and hence we consider an ideal scenario, in which the connection from  ${\rm U}_0$ to ${\rm U}_m$, $1\leq m\leq M$, is lossless.  This assumption  yields a lower bound on $Q_m$ as follows:
 \begin{align}
Q_m   
\geq&  \bar{\mu}_0\mathcal{E}_{h_0, s_0}\left\{  \left(1- \bar{F}_{\gamma}(\theta(h_0)) \right)^{M-m}  
  \left(\bar{F}_{\gamma}(a_{s_0,h_0}) \right)^{m}
\right\} ,
\end{align}
where $\bar{F}_{\gamma} (x)=1-e^{-\lambda_h x}$.  
For the case $E_m$, $m\geq 1$, we have $\theta(h_0)\geq 0$, which means that $|h_0|^2\geq   \epsilon_0P^{-1} $. 
In addition, $a_{s_0,h_0} = \min\{ \epsilon_s {P}^{-1}\beta^{-2}  |s_0|^{-2},  \theta(h_0)\}= \theta(h_0) $ implies the following
\begin{align}
 |h_0|^2\leq \epsilon_0\epsilon_s {P}^{-1}  |s_0|^{-2}+\epsilon_0P^{-1}  .
\end{align}

By applying the simplified CDF, $\bar{F}_{\gamma}(x)$, the lower bound on $Q_m$ can be expressed as follows:
\begin{align}
Q_m    
\geq&  \bar{\mu}_0  \int^{\infty}_{0}e^{-y} \int^{\frac{\epsilon_0\epsilon_s}{ {P} y}+\frac{\epsilon_0}{P}  }_{\frac{\epsilon_0}{P} } 
  \left(1-e^{-\lambda_h \theta(x)} \right)^{m} \\\nonumber & \times
 e^{-(M-m)\lambda_h \theta(x)}   
  \lambda_0e^{-\lambda_0x}dxdy \\\nonumber
  & +\bar{\mu}_0
 \int^{\infty}_{0} e^{-y} \left(1-e^{-\lambda_h \epsilon_s {P}^{-1}\beta^{-2}  y^{-1}} \right)^{m}\\\nonumber &\times \int_{\frac{\epsilon_0\epsilon_s}{ {P} y}+\frac{\epsilon_0}{P}  }^{\infty}
 e^{-(M-m)\lambda_h \theta(x)} 
   \lambda_0e^{-\lambda_0x}dx dy.
\end{align}
With some algebraic manipulations, the lower bound on $Q_m$ can be approximated at high SNR as follows: 
\begin{align}\label{es32}
Q_m    
\geq&    
    \bar{\mu}_0 \lambda_0  \sum^{m}_{p=0}{m \choose p} (-1)^p  \breve{\mu}_p^{-1} 
  \left[-  \frac{\breve{\mu}_p\epsilon_0\epsilon_s}{ {P} } \ln \frac{\breve{\mu}_p\epsilon_0\epsilon_s}{ {P} }  \right]
 \\\nonumber
  & +\bar{\mu}_0 \lambda_0 \breve{\mu}_0^{-1}\sum^{m}_{p=0}{m \choose p} (-1)^p 
\left(   \frac{4\tilde{\mu}_p}{ P } \ln \frac{4\tilde{\mu}_p}{ P } \right) 
\\\label{bd final1}
\rightarrow & \frac{1}{P\ln^{-1}P},
\end{align}
where the last approximation follows from the fact that each term in \eqref{es32} can be approximated as $\frac{1}{P\ln^{-1}P}$.  

{\it Remark 8:} Following the steps in the proof for Theorem \ref{theorem1} and also using \eqref{bd final1}, it is straightforward to show that the achievable diversity gain is one. In other words, the approximation obtained  in \eqref{bd final1} shows that the existence of virtual fast fading $|s_0|^2$ caps the diversity gain achieved by BAC-NOMA by one, even if the outage probability error floor can be discarded.

\begin{figure}[t] \vspace{-2em}
\begin{center}\subfigure[$R_0=0.1$ BPCU ]{\label{fig1a}\includegraphics[width=0.5\textwidth]{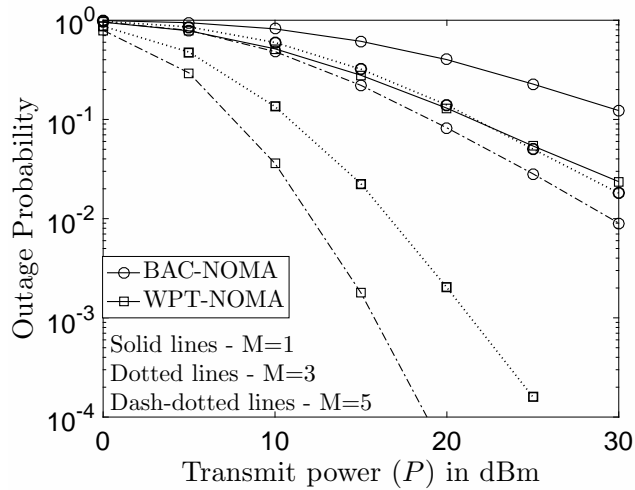}}
\subfigure[$R_0=2$ BPCU]{\label{fig1b}\includegraphics[width=0.5\textwidth]{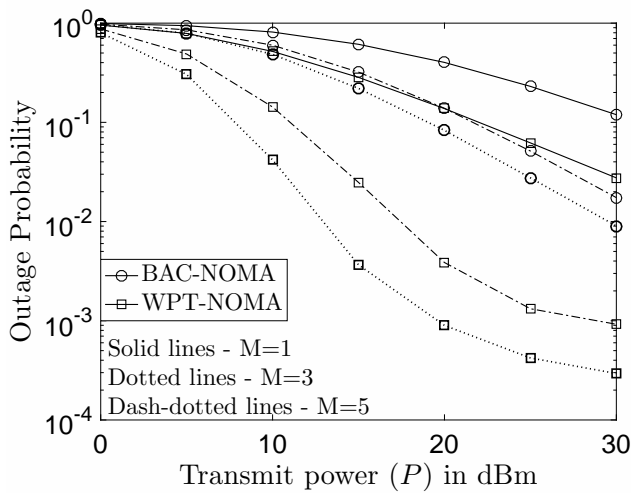}} \vspace{-1em}
\end{center} 
\caption{ Outage performance of  BAC-NOMA and WPT-NOMA.   $R_s=1.2$ bit per channel use (BPCU).  $d_h=d_0 = 50$ m and $d_g=5$ m. $\alpha=0.5$, $\beta=0.1$, and $\eta=0.1$.   \vspace{-1em} }\label{fig1}\vspace{-0.5em}
\end{figure}
\section{Simulation Results}\label{section simulation}
In this section, the performance of the two considered   transmission schemes, BAC-NOMA and WPT-NOMA, is investigated by using computer simulation results. For all the carried out simulations, we choose   $\phi=3.5$ and the noise power is $-94$ dBm. In Fig. \ref{fig1}, the outage performance achieved by WPT-NOMT and BAC-NOMA is studied with different choices of $R_0$. In Fig. \ref{fig1a}, the choice $R_0=0.1$ bits per channel use (BPCU) is used. With $R_0=0.1$ BPCU and $R_s=1.2$ BPCU, it is straightforward to verify that the condition $\bar{\epsilon}_0 \bar{\epsilon}_s<1$ holds. As indicated in Theorem \ref{theorem1}, if  $\bar{\epsilon}_0 \bar{\epsilon}_s<1$ holds, WPT-NOMA can avoid outage probability error floors,   which is consistent to the observations made from Fig. \ref{fig1a}. In addition, Fig. \ref{fig1a} shows that the slope of the outage probability curve for WPT-NOMA is increased when increasing $M$, which indicates that the diversity gain achieved by WPT-NOMA is increased by increasing $M$, an observation also consistent to   the conclusion made in Theorem \ref{theorem1}. In Fig. \ref{fig1b}, the choice $R_0=2$ BPCU is used, which leads to the violation of the condition  $\bar{\epsilon}_0 \bar{\epsilon}_s<1$. As a result, there are error floors for the outage probabilities achieved by WPT-NOMA, as shown in Fig. \ref{fig1b}. 

 \begin{figure}[t]\centering \vspace{-1em}
    \epsfig{file=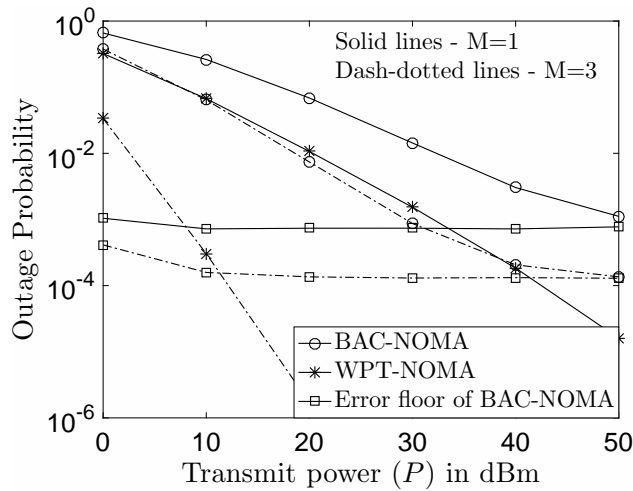, width=0.5\textwidth, clip=}\vspace{-0.5em}
\caption{ Illustration of the outage probability error floor of BAC-NOMA.   $R_0=0.1$ BPCU and $R_s=1.2$ BPCU.  $d_h=d_0 = 100$ m and $d_g=1$ m. $\alpha=0.5$, $\beta=0.1$, and $\eta=0.1$.   \vspace{-1em} }\label{fig3}\vspace{-0.5em}
\end{figure}

On the other hand,  the two figures in Fig. \ref{fig1} show that  BAC-NOMA always suffers outage probability error floors, which is due to the fact that hybrid SIC cannot be implemented in BAC-NOMA systems. In addition, the figures also demonstrate that the performance of BAC-NOMA can be improved by increasing $M$, i.e., inviting more delay-tolerant devices to participate in NOMA transmission is beneficial to improve reception reliability. But unlike WPT-NOMA, increasing $M$ does not change the slope of the outage probability curve for BAC-NOMA. It is worth to point out that  for the two considered choices of $R_0$,   WPT-NOMA can always realize a smaller outage probability than  BAC-NOMA, as shown in Fig. \ref{fig1}.  
 
 In Fig. \ref{fig3}, the outage probability error floor experienced by BAC-NOMA is studied, where the term in the legend, `Error Floor of BAC-NOMA',  refers to ${\rm P}(E_0)$. In order to clearly show the asymptotic behaviour of the outage probability, a larger transmit power range than those in Fig. \ref{fig1} is used. As can be observed from the figure, ${\rm P}(E_0)$ is   a tight lower bound on the outage probability, and it is   constant at high SNR, which implies that $E_0$ is the most damaging event and  is the cause for the error floor of the outage probability. Another important observation is that increasing    $M$ is useful to reduce the error floor, which confirms Lemma \ref{lemma1}. On the other hand,  WPT-NOMA does not suffer any outage probability error floor because the used target rate choices satisfy    $\bar{\epsilon}_0 \bar{\epsilon}_s<1$.  

\begin{figure}[t] \vspace{-2em}
\begin{center}\subfigure[Outage Probability]{\label{fig2a}\includegraphics[width=0.5\textwidth]{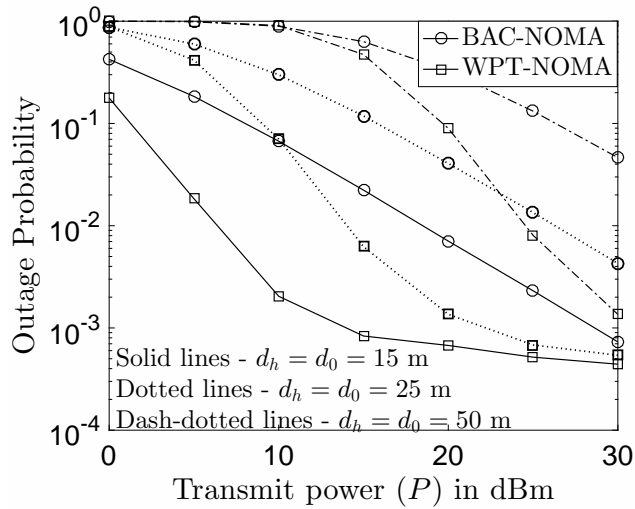}}
\subfigure[Ergodic Data Rate]{\label{fig2b}\includegraphics[width=0.5\textwidth]{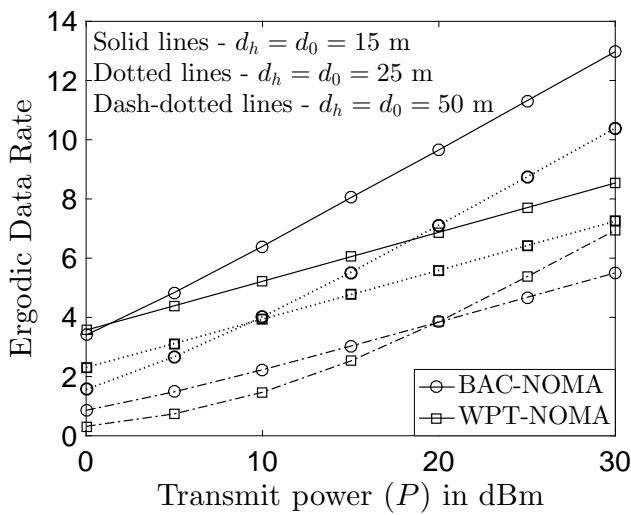}} \vspace{-1em}
\end{center} 
\caption{ Impact of path loss on the   performance of  BAC-NOMA and WPT-NOMA.   $M=5$, $R_0=2$ BPCU, $R_s=3$ BPCU.   $d_g=5$ m. $\alpha=0.5$, $\beta=0.1$, and $\eta=0.1$.   \vspace{-1em} }\label{fig2}\vspace{-0.5em}
\end{figure}

In Fig. \ref{fig2}, the impact of path loss on the performance of WPT-NOMA and BAC-NOMA is studied. In Fig. \ref{fig2a}, the outage probability is used as the metric for the performance evaluation, whereas the ergodic data rate is used as the metric in Fig. \ref{fig2b}. The two figures in Fig. \ref{fig2} show that the performance of the two NOMA schemes is degraded when path loss becomes more severe. This deteriorating  effect of path loss can be explained by using WPT-NOMA as an example. Increasing path loss does not only increase the attenuation of the signal strength, but also reduces the energy harvested at the delay-tolerant devices. For a similar reason, the performance of BAC-NOMA is also significantly affected by path loss. Therefore, the ideal applications of BAC-NOMA and WPT-NOMA   are   indoor   communication scenarios, e.g.,  the distances between the nodes are not large. We note that WPT-NOMA also exhibits outage probability error floors in Fig. \ref{fig2a}, since the condition $\bar{\epsilon}_0 \bar{\epsilon}_s<1$ does not hold, an observation consistent to the previous figures. In addition,  Fig. \ref{fig2a} shows that WPT-NOMA outperforms BAC-NOMA, if  the outage probability is used as the metric for performance evaluation,   which is also consistent to the previous numerical studies. However, Fig. \ref{fig2b} shows an interesting  result that BAC-NOMA can outperform WPT-NOMA if the ergodic rate is used as the performance metric, particularly at high SNR and with small path loss. One possible reason is   that WPT-NOMA relies on the time-switching WPT strategy, i.e.,  the first $\alpha T$ seconds are used for energy harvesting, and the remaining $(1-\alpha)T$ seconds are used for data transmission. On in other words, there is less time available for WPT-NOMA to transmit, whereas BAC-NOMA can carry out transmission continuously. 

In order to clearly demonstrate  the impact of $\alpha$ on the performance of WPT-NOMA, in  Fig. \ref{fig4}, different choices of $\alpha$ are used. In particular, $\alpha=0.1$ and $\alpha=0.9$ are a pair of choices of interest, as explained in the following.  The use of $\alpha=0.1$ means that the delay-tolerant devices use a small amount of time for energy harvesting and the majority time for data transmission, whereas $\alpha=0.9$ means that the majority time is used for energy harvesting. Fig. \ref{fig4} demonstrates that the choice of $\alpha=0.9$ results in the poorest performance among all the choices shown in the figure. This is due to the fact that there is not sufficient time for data transmission, even though a good amount of energy has been harvested and the delay-tolerant devices can use larger transmit powers than that in the case with $\alpha=0.1$.  It is worth pointing out that the choice of $\alpha=0.5$ yields the best performance among the choices shown in the figure.

\begin{figure}[t]\centering \vspace{-1em}
    \epsfig{file=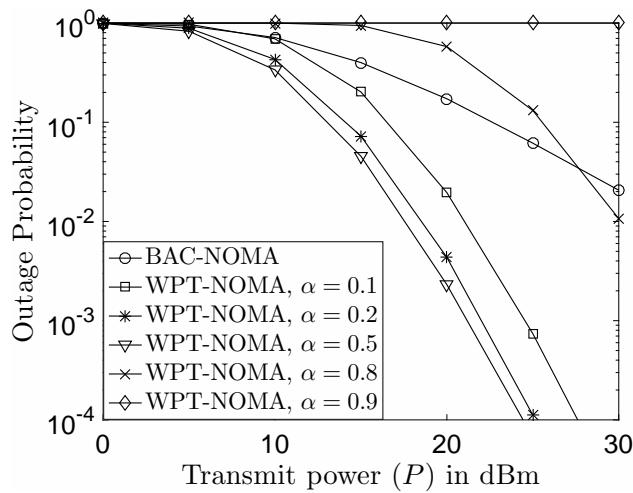, width=0.5\textwidth, clip=}\vspace{-0.5em}
\caption{ Impact of the choices of $\alpha$ on the performance of WPT-NOMA.    $R_0=0.1$ BPCU and $R_s=2$ BPCU.  $d_h=d_0 = 50$ m and $d_g=5$ m. $M=5$, $\beta=0.1$, and $\eta=0.1$.   \vspace{-1em} }\label{fig4}\vspace{-0.5em}
\end{figure}

\section{Conclusions}
 In this paper,  two energy and spectrally efficient transmission strategies,  namely WPT-NOMA and  BAC-NOMA,  were proposed by employing the energy and spectrum cooperation among the IoT devices. For the proposed WPT-NOMA scheme, hybrid SIC was used to improve reception reliability, and the developed analytical results demonstrate that   WPT-NOMA  can avoid outage probability error floors and realize the full diversity gain. Unlike WPT-NOMA, BAC-NOMA suffers from an outage probability error floor, and the asymptotic behaviour of this error floor was analyzed in the paper by applying EVT.  In addition, the effect of using one device's signal as the carrier signal  was studied, and its harmful impact on the diversity gain   was revealed. 
 
 We note that the provided simulation results show that the choice of $\alpha$ has a significant  impact on the performance of WPT-NOMA, and therefore an important direction for future research is to develop low-complexity algorithms for optimizing  $\alpha$. In addition, we note that the reason for BAC-NOMA to suffer the  outage probability error floor is due to the fact that hybrid SIC cannot be implemented. However, provided that ${\rm U}_n$, $1\leq n \leq M$, can carry out non-coherent detection, it is possible to apply hybrid SIC  to BAC-NOMA, which is another important direction for future research. 
 
 \appendices
 
 \section{Proof for Theorem \ref{theorem1}} \label{proof1}
 
The proof for the theorem can be divided to four steps, where the first three steps are to analyze the asymptotic behaviour of $T_0$, $T_m$, $1\leq m\leq M-1$, and $T_M$, respectively., and the last step is to study the overall diversity gain. 
\subsection{Asymptotic Study of $T_0$}
This section focuses on the high-SNR approximation of $T_0$ which can be rewritten as follows:
\begin{align} \nonumber
T_0 = &{\rm P}\left(    {\rm R}_{M} ^{WP,1}<R_s , |\mathcal{S}_2|=0\right)  
\\ \label{T 11}
= &{\rm P}\left(  \gamma_M<\frac{\bar{\epsilon}_s(P|h_0|^2+1)}{\eta P \bar{\alpha}} , \gamma_{1}>\tau(h_0) \right) .
\end{align}
 As can be observed from \eqref{T 11}, $T_0$ is a function of two order statistics,  $\gamma_1$ and $\gamma_M$, whose joint pdf is given by \cite{Arnoldbook}
\begin{align}
f_{\gamma_1,\gamma_M}(x,y) = \frac{M!}{(M-2)!}f_{\gamma}(x)f_{\gamma}(y)\left[F_{\gamma}(y)-F_{\gamma}(x)\right]^{M-2}.
\end{align}
Denote $T_{0|h_0}$ by the value of $T_0$ when $h_0$ is treated as a constant. Therefore,  $T_{0|h_0}$ can be expressed as follows:
\begin{align} \nonumber
T_{0|h_0}  =&  \frac{M!}{(M-2)!}\int^{\frac{\bar{\epsilon}_s(P|h_0|^2+1)}{\eta P \bar{\alpha}}}_{\tau(h_0)}f_{\gamma}(x)\int^{\frac{\bar{\epsilon}_s(P|h_0|^2+1)}{\eta P \bar{\alpha}}}_{x}f_{\gamma}(y)  \left[F_{\gamma}(y)-F_{\gamma}(x)\right]^{M-2}dydx
\\\nonumber
=&  \frac{M!}{(M-1)!}\int^{\frac{\bar{\epsilon}_s(P|h_0|^2+1)}{\eta P \bar{\alpha}}}_{\tau(h_0)}f_{\gamma}(x)
  \left[F_{\gamma}\left( \frac{\bar{\epsilon}_s(P|h_0|^2+1)}{\eta P \bar{\alpha}}\right)-F_{\gamma}(x)\right]^{M-1}dx.
\end{align}
$T_{0|h_0} $ can be further simplified  as follows:
\begin{align} 
T_{0|h_0}  =&    \left[F_{\gamma}\left( \frac{\bar{\epsilon}_s(P|h_0|^2+1)}{\eta P \bar{\alpha}}\right)-F_{\gamma}(\tau(h_0))\right]^{M} .
\end{align}

Therefore, $T_0$ can be obtained by finding the expectation of $T_{0|h_0} $ with respect to $h_0$:
\begin{align}\nonumber
T_0 = & \mathcal{E}_{h_0}\left\{ T_{0|h_0} \right\}. 
\end{align}
We note that $\tau(h_0)$ can have different forms depending on the choice of $|h_0|^2$. In particular, 
  $\tau(h_0)= 0$  means 
\begin{align}
  \frac{|h_0|^2}{\bar{\epsilon}_0\eta \bar{\alpha}} -\frac{1}{\eta P\bar{\alpha}} \leq 0  ,
 \end{align}
 which requires
 \begin{align}\label{low bound h}
  |h_0|^2 \leq \frac{\bar{\epsilon}_0 }{  P }  . 
 \end{align}

For the case $\tau(h_0)\neq 0$, the probability shown in \eqref{T 11} requires  $\tau(h_0)<\frac{\bar{\epsilon}_s(P|h_0|^2+1)}{\eta P \bar{\alpha}} $. This hidden constraint imposes another constraint on $|h_0|^2$ as follows: 
\begin{align}  
 \frac{|h_0|^2}{\bar{\epsilon}_0 } -\frac{1}{  P }  <\frac{\bar{\epsilon}_s(P|h_0|^2+1)}{  P },
 \end{align}
 which can be explicitly expressed as follows:
 \begin{align} 
  |h_0|^2    <    \frac{\bar{\epsilon}_0 (1+\bar{\epsilon}_s) }{  P(1-\bar{\epsilon}_0 \bar{\epsilon}_s ) } .\label{upper bound h}
\end{align}

By using the constraints shown in \eqref{low bound h} and \eqref{upper bound h}, $T_1$ can be expressed as follows:
\begin{align} \label{tx001}
T_0 = & \lambda_0 \int^{\frac{\bar{\epsilon}_0}{P}}_{0}\left[F_{\gamma}\left( \frac{\bar{\epsilon}_s(Px+1)}{\eta P \bar{\alpha}}\right)-F_{\gamma}(0)\right]^{M}  e^{-\lambda_0 x} dx
\\\nonumber & +\lambda_0
\int_{\frac{\bar{\epsilon}_0}{P}}^{{\frac{\bar{\epsilon}_0 (1+\bar{\epsilon}_s) }{  P(1-\bar{\epsilon}_0 \bar{\epsilon}_s ) }} }e^{-\lambda_0 x} \left[F_{\gamma}\left( \frac{\bar{\epsilon}_s(Px+1)}{\eta P \bar{\alpha}}\right)-F_{\gamma}\left( \frac{x}{\bar{\epsilon}_0\eta \bar{\alpha}} -\frac{1}{\eta P\bar{\alpha}}   \right)\right]^{M}   dx.
\end{align}
We note that the upper bound on $|h_0|^2$, ${\frac{\bar{\epsilon}_0 (1+\bar{\epsilon}_s) }{  P(1-\bar{\epsilon}_0 \bar{\epsilon}_s ) }} $,  is crucial to remove outage probability error floors and realize the full diversity gain, as shown in the following.  

In particular,   one can observe that both $ \frac{\bar{\epsilon}_s(Px+1)}{\eta P \bar{\alpha}}$ and $ \frac{x}{\bar{\epsilon}_0\eta \bar{\alpha}} -\frac{1}{\eta P\bar{\alpha}} $ go to zero for $P\rightarrow \infty$ in the two integrals considered in \eqref{tx001}. Therefore, the parameters of the Bessel functions in $T_0$ go to zero for $P\rightarrow \infty$.  Recall that   $xK_1(x) \approx  1+\frac{x^2}{2}\ln \frac{x}{2}$, for $x\rightarrow 0$ \cite{Dingkri04}. Therefore, the CDF of the unordered channel gain can be approximated as follows: 
\begin{align}\label{bessel approximation}
F_{\gamma}(x)   
=&1 - 2 \sqrt{ \lambda_h\lambda_g x}K_1\left(2 \sqrt{ \lambda_h\lambda_g x}\right)  \\\nonumber
\approx & 1 -   \left(  1+  \lambda_h\lambda_g x \ln (\lambda_h\lambda_g x)  \right)= - \lambda_h\lambda_g x \ln (\lambda_h\lambda_g x)  ,
\end{align}
for $x\rightarrow 0$. 
We note that for $x\rightarrow 0$, $\ln (\lambda_h\lambda_g x) <0$ and hence the approximation for $F_{\gamma}(x)   $ in \eqref{bessel approximation} is still positive.

Therefore, $T_0$ can be approximated at high SNR as follows:
\begin{align} \label{343}
T_0 \approx &  \int^{\frac{\bar{\epsilon}_0}{P}}_{0}\left[- \lambda_h\lambda_g  \frac{\bar{\epsilon}_s(Px+1)}{\eta P \bar{\alpha}} \ln \left(\lambda_h\lambda_g  \frac{\bar{\epsilon}_s(Px+1)}{\eta P \bar{\alpha}}\right)    \right]^{M}   dx  \lambda_0\\\nonumber &+\lambda_0
\int_{\frac{\bar{\epsilon}_0}{P}}^{{\frac{\bar{\epsilon}_0 (1+\bar{\epsilon}_s) }{  P(1-\bar{\epsilon}_0 \bar{\epsilon}_s ) }} }\left[ \lambda_h\lambda_g  \left(\frac{x}{\bar{\epsilon}_0\eta \bar{\alpha}} -\frac{1}{\eta P\bar{\alpha}} \right)     \ln \left(\lambda_h\lambda_g  \left(\frac{x}{\bar{\epsilon}_0\eta \bar{\alpha}} -\frac{1}{\eta P\bar{\alpha}} \right)  \right)
\right.
\\\nonumber &\left. 
-  \lambda_h\lambda_g  \frac{\bar{\epsilon}_s(Px+1)}{\eta P \bar{\alpha}} \ln \left(\lambda_h\lambda_g  \frac{\bar{\epsilon}_s(Px+1)}{\eta P \bar{\alpha}}\right) \right]^{M}   dx. 
\end{align}
In order to obtain a more insightful asymptotic expression of $T_0$, the expression in \eqref{343} can be rewritten as follows:
 \begin{align}\nonumber
T_0 \approx & \frac{\lambda_0}{P} \int^{ \bar{\epsilon}_0 }_{0}\left[- \lambda_h\lambda_g  \frac{\bar{\epsilon}_s(y+1)}{\eta P \bar{\alpha}} \ln \left(\lambda_h\lambda_g  \frac{\bar{\epsilon}_s(y+1)}{\eta P \bar{\alpha}}\right)    \right]^{M}   dy\\\nonumber
&+
\frac{\lambda_0}{P}\int_{ \bar{\epsilon}_0 }^{{\frac{\bar{\epsilon}_0 (1+\bar{\epsilon}_s) }{  (1-\bar{\epsilon}_0 \bar{\epsilon}_s ) }} }\left[ \frac{\lambda_h\lambda_g }{P} \left(\frac{y}{\bar{\epsilon}_0\eta \bar{\alpha}} -\frac{1}{\eta \bar{\alpha}} \right) \right. 
\\\nonumber &\left. \times\ln \left(\frac{\lambda_h\lambda_g}{P}  \left(\frac{y}{\bar{\epsilon}_0\eta \bar{\alpha}} -\frac{1}{\eta \bar{\alpha}} \right)  \right)
-  \lambda_h\lambda_g  \frac{\bar{\epsilon}_s(y+1)}{\eta P \bar{\alpha}} \ln \left(\lambda_h\lambda_g  \frac{\bar{\epsilon}_s(y+1)}{\eta P \bar{\alpha}}\right) \right]^{M}   dy
\\\nonumber 
=& \frac{\lambda_0}{P} \int^{ \bar{\epsilon}_0 }_{0}\left[- \frac{b_1(y)}{P} \ln \left(\frac{b_1(y)}{P}\right)    \right]^{M}   dy+
\frac{\lambda_0}{P}\int_{ \bar{\epsilon}_0 }^{{\frac{\bar{\epsilon}_0 (1+\bar{\epsilon}_s) }{  (1-\bar{\epsilon}_0 \bar{\epsilon}_s ) }} }\\ \label{two int}
&\times \left[ \frac{b_2(y)}{P}\ln \left(\frac{b_2(y)}{P} \right)
- \frac{b_1(y)}{P} \ln \left(\frac{b_1(y)}{P}\right) \right]^{M}   dy ,
\end{align}
where $y=Px$, $b_1(y)=\lambda_h\lambda_g  \frac{\bar{\epsilon}_s(y+1)}{\eta  \bar{\alpha}}$ and $b_2(y)=  \lambda_h\lambda_g  \left(\frac{y}{\bar{\epsilon}_0\eta \bar{\alpha}} -\frac{1}{\eta \bar{\alpha}} \right)$. It is important to point out that both $b_1(y)$ and $b_2(y)$ are constant and not functions of $P$. 

Denote the two integrals in \eqref{two int} by $\tilde{Q}_1$ and $\tilde{Q}_2$, respectively.    For $\tilde{Q}_1$,   the following approximation can be used:
 \begin{align}\label{eqx40}
  \frac{b_1(y)}{P}  \ln \left( \frac{b_1(y)}{P} \right) =&  \frac{b_1(y)}{P}  \left[
  \ln  b_1(y)  -\ln P   \right] \\  \underset{P\rightarrow \infty}{\approx }&
-  \frac{ b_1(y)}{P}   \ln P =-  \frac{b_1(y)}{P\ln^{-1}P} , \label{eqx41}
 \end{align}
 since $b_1(y)$ is finite and strictly larger than zero for the integral considered in $ \tilde{Q}_1$. Therefore, $Q_1$ can be approximated as follows:
 \begin{align}
 \tilde{Q}_1&\approx 
 \int^{ \bar{\epsilon}_0 }_{0}\left[\frac{b_1(y)}{P\ln^{-1}P}   \right]^{M}   dy  =   \frac{e_1}{P^{M}\ln^{-M}P} =     \mathcal{O}\left( \frac{1}{P^{M}\ln^{-M}P}\right),
 \end{align}
 where $\mathcal{O}$ denotes the approximation operation by omitting the constant multiplicative coefficient, and  the last approximation  follows from the fact that $e_1= \int^{ \bar{\epsilon}_0 }_{0}\left[ {b_1(y)}   \right]^{M}   dy$  is constant and not a function of $P$. 
 
 The approximation for $\tilde{Q}_2$ is more complicated since $b_2(y)$ can be zero for the considered integral and hence $\ln b_2(y)$ can be unbounded.  Unlike $\tilde{Q}_1$, $\tilde{Q}_2$ can be approximated as follows:
  \begin{align} \nonumber
\tilde{Q}_2=&\sum^{M}_{p=0}\frac{(-1)^p}{P^M}{M\choose p}
 \int_{ \bar{\epsilon}_0 }^{{\frac{\bar{\epsilon}_0 (1+\bar{\epsilon}_s) }{  (1-\bar{\epsilon}_0 \bar{\epsilon}_s ) }} }b_2(y) ^{M-p}b_1(y)^p \left[   \ln b_2(y) -\ln P 
 \right]^{M-p} \left[   \ln  b_1(y) - \ln P \right]^{p}   dy 
 \\\nonumber
 =&\sum^{M}_{p=0}\frac{(-1)^p}{P^M}
 \int_{ \bar{\epsilon}_0 }^{{\frac{\bar{\epsilon}_0 (1+\bar{\epsilon}_s) }{  (1-\bar{\epsilon}_0 \bar{\epsilon}_s ) }} }b_2(y) ^{M-p}b_1(y)^p \left( \sum^{M-p}_{i=0}(-1)^i{M-p\choose i}( \ln b_2(y))^{M-p-i} (\ln P)^i\right) \\ \label{dominant1} &\times \left(\sum^{p}_{j=0}{p\choose j}(-1)^j ( \ln  b_1(y))^{p-j}(\ln P)^j \right)    dy . 
\end{align}
At high SNR, the term with $(\ln P)^M$ is dominant, compared to the terms with $(\ln P)^m$, $m<M$, which means that  \eqref{dominant1}  can be further approximated as follows: 
   \begin{align} \nonumber
\tilde{Q}_2\approx  &\frac{(\ln P)^M }{P^M}\sum^{M}_{p=0} \underset{i+j=M}{\sum}(-1)^{p+i+j}{M-p\choose i}
 {p\choose j} \\    &\times  
\int_{ \bar{\epsilon}_0 }^{{\frac{\bar{\epsilon}_0 (1+\bar{\epsilon}_s) }{  (1-\bar{\epsilon}_0 \bar{\epsilon}_s ) }} }b_2(y) ^{M-p}b_1(y)^p( \ln  b_1(y))^{p-j} \\\nonumber &\times ( \ln b_2(y))^{M-p-i}    dy  =\mathcal{O}\left( \frac{1}{P^{M}\ln^{-M}P}\right).
\end{align}

Therefore, with $P\rightarrow \infty$, $T_0$ can be approximated as follows:
 \begin{align} \label{t00}
T_0 =\frac{\lambda_0}{P}\tilde{Q}_1+\frac{\lambda_0}{P}\tilde{Q}_2 =\mathcal{O}\left( \frac{1}{P^{M+1}\ln^{-M}P}\right). 
\end{align}

\subsection{Asymptotic Study of $T_m$, $1\leq m \leq M$}
This section is to   focus on    $T_m$, $1\leq m \leq M-1$, which can be expressed as follows:
\begin{align}\nonumber
T_m = &{\rm P}\left( {\rm R}_{m} ^{WP,2}<R_s,  {\rm R}_{M} ^{WP,1}<R_s , |\mathcal{S}_2|=m\right)  
\\\nonumber
= &{\rm P}\left( \gamma_m<\frac{\bar{\epsilon}_s}{\eta P\bar{\alpha}},\gamma_{m}<\tau(h_0),  \gamma_{m+1}>\tau(h_0), \gamma_M<\frac{\bar{\epsilon}_s(P|h_0|^2+1)}{\eta P \bar{\alpha}} \right) . 
\end{align}
 
 For the case of $1\leq m\leq M$, $\tau(h_0)\neq 0$, which means 
\begin{align}\label{low bound h2}
 \frac{|h_0|^2}{\bar{\epsilon}_0\eta \bar{\alpha}} -\frac{1}{\eta P\bar{\alpha}} >0,   
 \end{align}
 or equivalently $|h_0|^2 > \frac{\bar{\epsilon}_0 }{  P }  $. Furthermore, the requirement  $\tau(h_0)<\frac{\bar{\epsilon}_s(P|h_0|^2+1)}{\eta P \bar{\alpha}} $ leads to the constraint $ |h_0|^2    <    \frac{\bar{\epsilon}_0 (1+\bar{\epsilon}_s) }{  P(1-\bar{\epsilon}_0 \bar{\epsilon}_s ) }$, as discussed in \eqref{upper bound h}.

 Therefore, $T_m$ can be rewritten as follows:
\begin{align}\nonumber
T_m =& {\rm P}\left( \gamma_m<\frac{\bar{\epsilon}_s}{\eta P\bar{\alpha}}   , \gamma_M<\frac{\bar{\epsilon}_s(P|h_0|^2+1)}{\eta P \bar{\alpha}} , \right.
\\\nonumber &\left.\gamma_{m}<\frac{|h_0|^2}{\bar{\epsilon}_0\eta \bar{\alpha}} -\frac{1}{\eta P\bar{\alpha}}  , \gamma_{m+1}>\frac{|h_0|^2}{\bar{\epsilon}_0\eta \bar{\alpha}} -\frac{1}{\eta P\bar{\alpha}}  \right) 
\\\label{tm1}
=&{\rm P}\left( \gamma_m< b_{h_0}  ,  \gamma_{m+1}> \tau(h_0) , \gamma_M<a(h_0) \right) ,
\end{align}
where $a(h_0)=\frac{\bar{\epsilon}_s(P|h_0|^2+1)}{\eta P \bar{\alpha}} $ and $ b_{h_0}  = \min\left\{ \frac{\bar{\epsilon}_s}{\eta P\bar{\alpha}} , \tau(h_0)\right\}$.

As can be observed from \eqref{tm1}, $T_m$, $1\leq m \leq M-1$, is a function of three order statistics, $\gamma_m$, $\gamma_{m+1}$, and $\gamma_{M}$.  Recall that the joint pdf of three order statistics is given by \cite{Arnoldbook}
\begin{align}\label{joint pdf three}
&f_{\gamma_m,\gamma_{m+1},\gamma_{M}}(x,y,z) = c_m F_{\gamma}(x)^{m-1}\\\nonumber  &\times \left(F_{\gamma}(z) - F_{\gamma}(y)\right)^{M-m-2} f_{\gamma}(x)f_{\gamma}(y)f_{\gamma}(z),%
\end{align}
where $c_m=\frac{M!}{(m-1)!(M-m-2)!}$.

Denote $T_{m|h_0}$ by the value of  $T_m$ by assuming  that          $h_0$ is fixed.  By using the joint pdf in \eqref{joint pdf three},   $T_{m|h_0} $ can be expressed as follows:
\begin{align} 
T_{m|h_0} =&  {\rm P}\left( \gamma_m< b_{h_0}  ,  \gamma_{m+1}> \tau(h_0) , \gamma_M<a(h_0) \right) 
\\\nonumber =&  c_m \int^{b_{h_0} }_{0}F_{\gamma}(x)^{m-1} f_{\gamma}(x)dx \int^{a(h_0)}_{ \tau(h_0) }f_{\gamma}(y) \int^{a(h_0)}_{y}\left(F_{\gamma}(z) - F_{\gamma}(y)\right)^{M-m-2} f_{\gamma}(z)dz.
\end{align}
 
By using the property of CDFs, $T_{m|h_0} $ can be more explicitly expressed as follows: 
\begin{align} \nonumber
T_{m|h_0} =&    \bar{ c}_m F_{\gamma}(b_{h_0} )^{m} \int^{a(h_0)}_{ \tau(h_0) }\left[\left(F_{\gamma}(a(h_0)) - F_{\gamma}(y)\right)^{M-m-1}  -\left(F_{\gamma}(y) - F_{\gamma}(y)\right)^{M-m-1}  \right]f_{\gamma}(y)dy
\\ 
=&    \bar{ c}_m F_{\gamma}(b_{h_0} )^{m} \int^{a(h_0)}_{ \tau(h_0) }   \left[F_{\gamma}(a(h_0)) - F_{\gamma}(y)\right]^{M-m-1}   f_{\gamma}(y) dy,
\end{align}
where $\bar{c}_m=\frac{M!}{m!(M-m-1)!}$. The expression of  $T_{m|h_0} $ can be further simplified as follows:
\begin{align} \nonumber
T_{m|h_0} =&    \tilde{ c}_m F_{\gamma}(b_{h_0} )^{m}  \left(   \left[F_{\gamma}(a(h_0)) - F_{\gamma}(\tau(h_0))\right]^{M-m} \right.\\\nonumber
&\left. - \left[F_{\gamma}(a(h_0)) - F_{\gamma}(a(h_0))\right]^{M-m}  \right)
\\ \label{m general}
=&    \tilde{ c}_m F_{\gamma}(b_{h_0} )^{m}   \left[F_{\gamma}(a(h_0)) - F_{\gamma}(\tau(h_0))\right]^{M-m}     ,
\end{align}
where  $\tilde{c}_m=\frac{M!}{m!(M-m)!}$. 

$T_m $ can be obtained by calculating the expectation of $T_{m|h_0} $ with   respect of $|h_0|^2$ as follows: 
 \begin{align}
T_m   
= &  \mathcal{E}_{h_0}\left\{  T_{m|h_0} \right\}
\\\nonumber
=&  \tilde{ c}_m   \mathcal{E}_{h_0}\left\{  F_{\gamma}(b_{h_0} )^{m}   \left[F_{\gamma}(a(h_0)) - F_{\gamma}(\tau(h_0))\right]^{M-m}   \right\}.
\end{align}

Recall that $b_{h_0}=\tau(h_0)$ if the constraint $\frac{\bar{\epsilon}_s}{\eta P\bar{\alpha}}   >\frac{|h_0|^2}{\bar{\epsilon}_0\eta \bar{\alpha}} -\frac{1}{\eta P\bar{\alpha}}$ is satisfied, which imposes  the following constraint on $ |h_0|^2$: 
\begin{align}\label{cons3}
   |h_0|^2 <\frac{\bar{\epsilon}_0 (1+\bar{\epsilon}_s)}{  P }    . 
\end{align}
Therefore, $T_m  $ can be more explicitly expressed as follows:
 \begin{align}
T_m   \label{integral tm}
= &    \tilde{ c}_m \lambda_0 \int^{\frac{\bar{\epsilon}_0 (1+\bar{\epsilon}_s)}{  P }  }_{\frac{\bar{\epsilon}_0 }{  P }  }     \left(   \left[F_{\gamma}(a(x)) - F_{\gamma}(\tau(x))\right]^{M-m}    \right) \\\nonumber &\times F_{\gamma}(\tau(x) )^{m}e^{-\lambda_0 x}dx +\tilde{ c}_m\lambda_0F_{\gamma}\left(\frac{\bar{\epsilon}_s}{\eta P\bar{\alpha}} \right)^{m}   \\\nonumber &\times \int_{\frac{\bar{\epsilon}_0 (1+\bar{\epsilon}_s)}{  P }  }^{\frac{\bar{\epsilon}_0 (1+\bar{\epsilon}_s) }{  P(1-\bar{\epsilon}_0 \bar{\epsilon}_s ) }}    \left(   \left[F_{\gamma}(a(x)) - F_{\gamma}(\tau(x))\right]^{M-m}    \right) e^{-\lambda_0 x}dx,
\end{align}
where the constraints on $|h_0|^2$ shown in  \eqref{upper bound h}, \eqref{low bound h2}  and \eqref{cons3} have been used.

We note that for the integrals considered in \eqref{integral tm},  $ \tau(x) \rightarrow 0$ for $P\rightarrow \infty$, which can be explained in the following. Recall that  
\begin{align}
 \tau(x) = \frac{x}{\bar{\epsilon}_0\eta \bar{\alpha}} -\frac{1}{\eta P\bar{\alpha}}  .
\end{align}
For the integrals considered in \eqref{integral tm}, ${\frac{\bar{\epsilon}_0 }{  P }  } \leq x\leq {\frac{\bar{\epsilon}_0 (1+\bar{\epsilon}_s)}{  P }  } $ and ${\frac{\bar{\epsilon}_0 (1+\bar{\epsilon}_s)}{  P }  }\leq x\leq {\frac{\bar{\epsilon}_0 (1+\bar{\epsilon}_s) }{  P(1-\bar{\epsilon}_0 \bar{\epsilon}_s ) }} $.  Therefore, indeed
  $x\rightarrow 0$ for $P\rightarrow \infty$, which means that $ \tau(x) \rightarrow 0$. Similarly, for the integrals considered in \eqref{integral tm}, the following approximation also holds 
  \begin{align}
a(x)= \frac{\bar{\epsilon}_sx}{\eta  \bar{\alpha}}+\frac{\bar{\epsilon}_s }{\eta P \bar{\alpha}}\underset{P\rightarrow\infty}{\longrightarrow} 0.
\end{align}
 
 By using these asymptotic behaviours of $\tau(x)$ and $a(x)$,  the probability $T_m   $ can be approximated as follows:
\begin{align}\nonumber
T_m   
\approx &    \tilde{ c}_m \lambda_0 \int^{\frac{\bar{\epsilon}_0 (1+\bar{\epsilon}_s)}{  P }  }_{\frac{\bar{\epsilon}_0 }{  P }  }    \left[- \lambda_h\lambda_g \tau(x)  \ln (\lambda_h\lambda_g \tau(x) ) \right]^{m}\left[  \lambda_h\lambda_g \tau(x)\right.  \\\nonumber &\times \left. \ln (\lambda_h\lambda_g \tau(x)) - \lambda_h\lambda_g a(x) \ln (\lambda_h\lambda_g a(x))  \right]^{M-m}     dx
\\\nonumber
&+\tilde{ c}_m\lambda_0\left[ - \lambda_h\lambda_g \frac{\bar{\epsilon}_s}{\eta P\bar{\alpha}} \ln \left(\lambda_h\lambda_g \frac{\bar{\epsilon}_s}{\eta P\bar{\alpha}}\right)\right]  ^{m}   \\\nonumber
&\times \int_{\frac{\bar{\epsilon}_0 (1+\bar{\epsilon}_s)}{  P }  }^{\frac{\bar{\epsilon}_0 (1+\bar{\epsilon}_s) }{  P(1-\bar{\epsilon}_0 \bar{\epsilon}_s ) }}    \left[  \lambda_h\lambda_g \tau(x) \ln (\lambda_h\lambda_g \tau(x))  - \lambda_h\lambda_g a(x) \ln (\lambda_h\lambda_g a(x))  \right]^{M-m} dx,
\end{align}
for $P\rightarrow \infty$. 

Define $
 \bar{\tau}(h_0) = P\lambda_h\lambda_g\tau(h_0) $ and $
\bar{a}(h_0)=  P\lambda_h\lambda_g a(h_0)$.  
Therefore, $T_m   $ can be expressed as follows: 
\begin{align}
T_m   
\approx &    \tilde{ c}_m\lambda_0  \int^{\frac{\bar{\epsilon}_0 (1+\bar{\epsilon}_s)}{  P }  }_{\frac{\bar{\epsilon}_0 }{  P }  }    \left[-  \frac{ \bar{\tau}(x)}{P}  \ln \left( \frac{\bar{ \tau}(x)}{P} \right) \right]^{m}  \\\nonumber &\times \left[  \frac{\bar{ \tau}(x)}{P} \ln \left(\frac{\bar{ \tau}(x)}{P}\right) - \frac{\bar{ a}(x)}{P} \ln \left(\frac{\bar{ a}(x)}{P}\right)  \right]^{M-m}     dx
\\\nonumber
&+\tilde{ c}_m\lambda_0\left[ -  \frac{\lambda_h\lambda_g\bar{\epsilon}_s}{\eta P\bar{\alpha}} \ln \left(\frac{\lambda_h\lambda_g \bar{\epsilon}_s}{\eta P\bar{\alpha}}\right)\right]  ^{m} \int_{\frac{\bar{\epsilon}_0 (1+\bar{\epsilon}_s)}{  P }  }^{\frac{\bar{\epsilon}_0 (1+\bar{\epsilon}_s) }{  P(1-\bar{\epsilon}_0 \bar{\epsilon}_s ) }}   \\\nonumber
&\times    \left[ \frac{\bar{ \tau}(x)}{P} \ln \left( \frac{ \bar{\tau}(x)}{P}\right) - \frac{\bar{ a}(x)}{P} \ln \left( \frac{\bar{ a}(x)}{P}\right)  \right]^{M-m} dx.
\end{align}

In order to obtain a more insightful asymptotic expression, we substitute    the following three parameters,  $y=Px$,
\begin{align}
 \tilde{\tau}(y) = \lambda_h\lambda_g\left(\frac{y}{\bar{\epsilon}_0\eta \bar{\alpha}} -\frac{1}{\eta \bar{\alpha}} \right) ,
\end{align}
and
\begin{align}
\tilde{a}(y)= \lambda_h\lambda_g\left(\frac{\bar{\epsilon}_sy}{\eta  \bar{\alpha}}+\frac{\bar{\epsilon}_s }{\eta  \bar{\alpha}}\right),
\end{align}
into the expression of $T_m$, which yields the following expression: 
\begin{align}
T_m   
\approx &    \frac{\tilde{ c}_m \lambda_0}{P} \int^{ \bar{\epsilon}_0 (1+\bar{\epsilon}_s)   }_{ \bar{\epsilon}_0    }    \left[-  \frac{ \tilde{\tau}(y)}{P}  \ln \left( \frac{\tilde{ \tau}(y)}{P} \right) \right]^{m}  \\\nonumber &\times \left[  \frac{\tilde{ \tau}(y)}{P} \ln \left(\frac{\tilde{ \tau}(y)}{P}\right) - \frac{\tilde{ a}(y)}{P} \ln \left(\frac{\tilde{ a}(y)}{P}\right)  \right]^{M-m}     dy
\\\nonumber
&+\frac{\tilde{ c}_m\lambda_0}{P}\left[ -  \frac{\lambda_h\lambda_g\bar{\epsilon}_s}{\eta P\bar{\alpha}} \ln \left(\frac{\lambda_h\lambda_g \bar{\epsilon}_s}{\eta P\bar{\alpha}}\right)\right]  ^{m}  \int_{\bar{\epsilon}_0 (1+\bar{\epsilon}_s)  }^{\frac{\bar{\epsilon}_0 (1+\bar{\epsilon}_s) }{  (1-\bar{\epsilon}_0 \bar{\epsilon}_s ) }}    \left[ \frac{\tilde{ \tau}(y)}{P} \ln \left( \frac{ \tilde{\tau}(y)}{P}\right) - \frac{\tilde{ a}(y)}{P} \ln \left( \frac{\tilde{ a}(y)}{P}\right)  \right]^{M-m} dy. 
\end{align}
It is important to point out that   both $ \tilde{\tau}(y) $ and $\tilde{a}(y)$ are constant and not functions of $P$.
 By using the steps similar to those to obtain the approximation of $T_0$,      $T_m $ can be approximated   as follows:  \begin{align} \label{tm xx}
T_m   
= &   \mathcal{O}\left( \frac{1}{P^{M+1}\ln^{-M}P}\right) .
\end{align}

\subsection{Asymptotic Study of $T_M$}
For the special case $T_M$, we first recall that $T_M$ can be expressed as follows:
\begin{align}\nonumber
T_M = &{\rm P}\left( {\rm R}_{M} ^{WP,2}<R_s,  {\rm R}_{M} ^{WP,1}<R_s , |\mathcal{S}_2|=M\right)  
\\ \label{TM x}
= &{\rm P}\left( \gamma_M<\frac{\bar{\epsilon}_s}{\eta P\bar{\alpha}},   \gamma_{M}<\tau(h_0) \right) .
\end{align}

By using the marginal pdf of the largest order statistics, $T_M$ can be   be straightforwardly  expressed as follows:
\begin{align} \label{TM re}
T_M = & \mathcal{E}_{h_0}\left\{ F_{\gamma}\left(\min\left\{\frac{\bar{\epsilon}_s}{\eta P\bar{\alpha}},\tau(h_0)\right\}\right)^{M} \right\}.
\end{align}
As can be observed from \eqref{TM x}, $T_M$ is a function of $\gamma_M$ only, which is different from $T_m$, $1\leq m\leq M-1$. It is important to point out that the constraint of $|h_0|^2$ shown in \eqref{upper bound h} does not exist for  $T_M$. This causes the reduction of    the diversity gain from $M+1$ to $M$, as shown in the following.   $T_M  $ can be more explicitly expressed as follows:
 \begin{align}
T_M   
= &   \underset{T_{M,1}}{\underbrace{ \tilde{ c}_m \lambda_0 \int^{\frac{\bar{\epsilon}_0 (1+\bar{\epsilon}_s)}{  P }  }_{\frac{\bar{\epsilon}_0 }{  P }  }       F_{\gamma}(\tau(x) )^{M}e^{-\lambda_0 x}dx}}  \\ \label{eq52} &+\tilde{ c}_m\lambda_0
 \underset{T_{M,2}}{\underbrace{ F_{\gamma}\left(\frac{\bar{\epsilon}_s}{\eta P\bar{\alpha}} \right)^{M}  }}  \underset{T_{M,3}}{\underbrace{ \int_{\frac{\bar{\epsilon}_0 (1+\bar{\epsilon}_s)}{  P }  }^{\infty}   e^{-\lambda_0 x}dx}}.
\end{align}
 By following steps similar to those to analyze $T_m$, $1\leq m\leq M-1$, it is straightforward to show that   $T_{M,1}   
=\mathcal{O}\left(     \frac{1 }{P^{M+1} \ln^{-M}P} \right)$ and $T_{M,2}   
=\mathcal{O}\left(      \frac{1 }{P^{M} \ln^{-(M-1)}P} \right)$.

What makes the high SNR behaviour of $T_M$ different from those of $T_m$, $0\leq m\leq M-1$, is $T_{M,3}$.  It is important to point out that  the upper end of the integral range of $T_{M,3}$ is  $\infty$, instead of a value which goes to zero for $P\rightarrow \infty$. As a result, $\lambda_0T_{M,3}=   e^{-\lambda_0 \frac{\bar{\epsilon}_0 (1+\bar{\epsilon}_s)}{  P }  }\underset{P\rightarrow \infty}{\longrightarrow} 1$, instead of $\frac{1}{P}$. Therefore, $T_M$ can be approximated at high SNR as follows:
 \begin{align} \label{tm xx2}
T_M  
= &  \mathcal{O}\left(  \frac{1 }{P^{M} \ln^{-(M-1)}P}\right)  . 
\end{align}

\subsection{Overall High-SNR Approximation}
By substituting  \eqref{t00}, \eqref{tm xx} and \eqref{tm xx2} in \eqref{wp outage}, we can conclude that the overall outage probability can be approximated as follows:
\begin{align} \label{ovx}
{\rm P}^{WP}  = &  \mathcal{O}\left(  \frac{1 }{P^{M} \ln^{-(M-1)}P} \right),
\end{align}
for $P\rightarrow \infty$.  \eqref{ovx} indicates that $T_M$ is the most dominant term   in \eqref{wp outage} at high SNR. 

The diversity gain achieved by WPT-NOMA can be obtained as follows: 
\begin{align}
d=& \underset{P\rightarrow \infty}{\lim}-\frac{\log {\rm P}^{WP} }{\log P} =\underset{P\rightarrow \infty}{\lim}\frac{\log \left(P^{M} \ln^{-(M-1)}P\right)}{\log P}
\\\nonumber
=&\underset{P\rightarrow \infty}{\lim}\left[ \frac{\log  P^{M} }{\log P}-
\frac{\log  \ln^{M-1}P }{\log P}\right].
\end{align}
The following limit holds at high SNR
\begin{align}
 \underset{P\rightarrow \infty}{\lim} 
\frac{\log  \ln^{M-1}P }{\log P}=& \underset{P\rightarrow \infty}{\lim} 
\frac{\log e \ln\left(  \ln^{M-1}P \right)}{\log e \ln P}= \underset{P\rightarrow \infty}{\lim}  \frac{M-1}{\ln P} =0,
\end{align}
where L'Hospital's rule is used. Therefore, the diversity gain achieved by WPT-NOMA can be obtained as follows:
\begin{align}
d=&\underset{P\rightarrow \infty}{\lim}  \frac{\log  P^{M} }{\log P} =M,
\end{align}
and the theorem is proved.

\section{Proof for Lemma \ref{lemma1}} \label{lemma1proof}

In order to study the asymptotic behaviour of  $ {\rm P}\left( E_0\right) $, EVT is applied in the following. Recall that the limiting CDF of the smallest order statistics should follow one of the three distributions,  namely the Frech\'et type, the modified Weibull type and the extreme value CDF \cite[Theorem 8.3.5]{Arnoldbook}. For the considered order statistics, $\gamma_1$,  the modified Weibull type is   applicable  as explained in the following.  

Denote $F^{-1}_{\gamma}(a)$ by the inverse function of the CDF of the unordered channel gain, i.e., $F_{\gamma}\left(F^{-1}_{\gamma}(a)\right)=a$.  The first condition to show that the considered CDF is the modified Weibull type of EVT is that  $F^{-1}_{\gamma}(0)$ should be finite \cite[Theorem 8.3.6]{Arnoldbook}.  For the considered CDF, we have $F^{-1}_{\gamma}(0)=0$ which is indeed finite. The second condition   is to   show whether the following limitation exists  
\begin{align}
\underset{\epsilon\rightarrow 0^+}{\lim}\frac{F_{\gamma}\left(F^{-1}_{\gamma}(0)+\epsilon x\right)}{F_{\gamma}\left(F^{-1}_{\gamma}(0)+\epsilon \right)}=x^{\breve{\alpha}},
\end{align}
for all $x>0$, where $\breve{\alpha}$ denotes  a constant parameter. 

For the considered CDF, the limitation can be expressed as follows:
\begin{align}\label{evt1}
&\underset{\epsilon\rightarrow 0^+}{\lim}\frac{F_{\gamma}\left(F^{-1}_{\gamma}(0)+\epsilon x\right)}{F_{\gamma}\left(F^{-1}_{\gamma}(0)+\epsilon \right)} 
\\\nonumber =&\underset{\epsilon\rightarrow 0^+}{\lim}\frac{F_{\gamma}\left( \epsilon x\right)}{F_{\gamma}\left( \epsilon \right)} = \underset{\epsilon\rightarrow 0^+}{\lim}\frac{1 - 2 \sqrt{ \lambda_h\lambda_g \epsilon x}K_1\left(2 \sqrt{ \lambda_h\lambda_g \epsilon x}\right) }{1 - 2 \sqrt{ \lambda_h\lambda_g \epsilon}K_1\left(2 \sqrt{ \lambda_h\lambda_g \epsilon}\right) } .
\end{align}
 
Note that in \eqref{evt1}, $x$ is constant, and the limitation is with respect to $\epsilon$. When $\epsilon\rightarrow 0$, the approximation in \eqref{bessel approximation} can be applied and the limitation can be obtained as follows:
\begin{align}\label{evt2}
&\underset{\epsilon\rightarrow 0^+}{\lim}\frac{F_{\gamma}\left(F^{-1}_{\gamma}(0)+\epsilon x\right)}{F_{\gamma}\left(F^{-1}_{\gamma}(0)+\epsilon \right)} 
=  \underset{\epsilon\rightarrow 0^+}{\lim}\frac{- \lambda_h\lambda_g \epsilon x \ln (\lambda_h\lambda_g \epsilon x) }{- \lambda_h\lambda_g \epsilon \ln (\lambda_h\lambda_g \epsilon) } =\underset{\epsilon\rightarrow 0^+}{\lim}\frac{ x \ln (\lambda_h\lambda_g \epsilon x) }{ \ln (\lambda_h\lambda_g \epsilon) } .
\end{align}
By applying  L'Hospital's rule, the limitation can be obtained as follows:
\begin{align}\label{evt3}
&\underset{\epsilon\rightarrow 0^+}{\lim}\frac{F_{\gamma}\left(F^{-1}_{\gamma}(0)+\epsilon x\right)}{F_{\gamma}\left(F^{-1}_{\gamma}(0)+\epsilon \right)} 
=\underset{\epsilon\rightarrow 0^+}{\lim}\frac{ x \frac{\lambda_h\lambda_g  x}{\lambda_h\lambda_g \epsilon x} }{ \frac{\lambda_h\lambda_g }{\lambda_h\lambda_g \epsilon} } =x,
\end{align}
which means that $\breve{\alpha}=1$ for the considered order statistics.

As a result, the smallest channel gain will follow the modified Weibull type with $\breve{\alpha}=1$, i.e.,
\begin{align}\label{ab}
\frac{\gamma_1-a_m}{b_m} \sim G_2^*(x;\breve{\alpha}) ,
\end{align}
where $G_2^*(x;\breve{\alpha}) $ denotes the modified Weibull distribution:
\begin{align}
G_2^*(x;\breve{\alpha}) \triangleq  1-G_2(-x;\breve{\alpha}) = 1- e^{-x},
\end{align}
and $G_2(x;\breve{\alpha})$ denotes the Weibull  distribution defined as follows:
\begin{align}
G_2(x;\breve{\alpha}) \triangleq  \left\{\begin{array}{ll}e^{-(-x)^{\breve{\alpha}}}, &x<0\\
1,&x\geq 0
\end{array}\right..
\end{align}
The two parameters in \eqref{ab}, $a_m$ and $b_m$ are given by
\begin{align}\label{am}
a_m \triangleq  F^{-1}_{\gamma}(0)=0,
\end{align}
and 
\begin{align}\label{bmxx}
b_m \triangleq F^{-1}_{\gamma}\left(\frac{1}{M}\right) - F^{-1}_{\gamma}(0)=F^{-1}_{\gamma}\left(\frac{1}{M}\right).
\end{align}
The challenging step is to find an explicit expression of $b_m$, which can be obtained   by solving the following equation: 
\begin{align} \label{eqx3}
 1 - 2 \sqrt{ \lambda_h\lambda_g b_m}K_1\left(2 \sqrt{ \lambda_h\lambda_g b_m}\right) = \frac{1}{M} .
 \end{align}
For $M\rightarrow \infty$, we have $\frac{1}{M}\rightarrow 0$ and hence $b_m\rightarrow 0$. Because  $b_m\rightarrow 0$, the use of the approximation in \eqref{bessel approximation} can be used to simplify  the   equation \eqref{eqx3}  as follows:
\begin{align} \label{lambert1}
 - \lambda_h\lambda_g b_m \ln (\lambda_h\lambda_g b_m)  =\frac{1}{M}.
\end{align}
In order to apply the Lambert W function,   \eqref{lambert1} needs to be written as follows:
\begin{align} \label{lambert2}
    - \frac{1}{M}  &=-\frac{1}{M \lambda_h\lambda_g b_m}e^{-\frac{1}{M\lambda_h\lambda_g b_m}},
 \end{align}
 which means that the solution of \eqref{lambert2} can be expressed as follows: 
 \begin{align} \label{lambert3}
  -\frac{1}{M \lambda_h\lambda_g b_m} &=W\left(  - \frac{1}{M} \right),
 \end{align}
 or equivalently
  \begin{align} \label{lambert4}
 b_m &=-\frac{1}{ M \lambda_h\lambda_gW\left(  - \frac{1}{M} \right)},
 \end{align}
 where $W(\cdot)$ denotes the  Lambert W function.
 
 Because $- \frac{1}{M} $ is negative, there are two solutions for $W\left(  - \frac{1}{M} \right)$, namely $W_0\left(  - \frac{1}{M} \right)$ and $W_{-1}\left(  - \frac{1}{M} \right)$ \cite{6559999}. Recall that $W_0(x)\rightarrow 0$ for $x\rightarrow 0$, which means that $ b_m =-\frac{1}{ M \lambda_h\lambda_gW_0\left(  - \frac{1}{M} \right)}\rightarrow \infty$ for  $M\rightarrow \infty$. This is contradicted to  \eqref{bmxx} which indicates that $b_m\rightarrow 0$ for $M\rightarrow \infty$. Therefore, $W_0\left(  - \frac{1}{M} \right)$  is not the solution of the considered case, and we are interested the other branch, $W_{-1}\left(  - \frac{1}{M} \right)$. Recall that $W_{-1}\left( x\right)$ can be bounded as follows: \cite{6559999}
 \begin{align}
 -1-\sqrt{2u}-u< W_{-1}\left( -e^{-u-1}\right)<-1-\sqrt{2u} -\frac{2}{3}u,\text{ for } u>0.
 \end{align} 
 
By applying the bounds,  $W_{-1}\left(  - \frac{1}{M} \right)$ can be bounded as follows:
  \begin{align}
 \ln \frac{1}{M}< W_{-1}\left(  - \frac{1}{M} \right)< \frac{2}{3}\ln \frac{1}{M},
 \end{align}
 which yields the following approximation:
  \begin{align}
 W_{-1}\left(  - \frac{1}{M} \right)=- \mathcal{O}( \ln M).
 \end{align}
 Therefore, $b_m$ can be approximated as follows:
   \begin{align} \label{lambert5}
 b_m &=\frac{1}{  \lambda_h\lambda_g M \mathcal{O}( \ln M)}.
 \end{align}
 
By applying \eqref{am} and \eqref{lambert5} to \eqref{ab}, we have $ \frac{\gamma_1 }{b_m}\sim e^{-x} $ and the limiting CDF of the smallest channel gain is given by 
 \begin{align}
  F_{\gamma_1}(y) = 1-e^{ {y}{M \lambda_h\lambda_gW\left(  - \frac{1}{M} \right)}},
 \end{align}
 and the corresponding pdf is given by $  f_{\gamma_1}(y) = {M \lambda_h\lambda_gW\left(  - \frac{1}{M} \right)}e^{ {y}{M \lambda_h\lambda_gW\left(  - \frac{1}{M} \right)}}
$.
 
By using  this pdf, $ {\rm P}\left( E_0\right)  $ can be expressed as follows: 
\begin{align}\nonumber
 {\rm P}\left( E_0\right)  
=& \int^{\infty}_{0}\left(e^{-\lambda_0\epsilon_0P^{-1}  }- e^{ -\lambda_0(  \beta^2 \epsilon_0x+  \epsilon_0 P^{-1})} \right)f_{\gamma_1}(x)dx   +1 - e^{-\lambda_0\epsilon_0P^{-1} }
\\\nonumber   
 \approx &1 +  \frac{M \lambda_h\lambda_gW\left(  - \frac{1}{M} \right)  }{\lambda_0   \beta^2 \epsilon_0- M \lambda_h\lambda_gW\left(  - \frac{1}{M} \right)}  ,
\end{align}
which can be approximated as follows:
\begin{align}\nonumber
 {\rm P}\left( E_0\right)  
 \approx &1 -  \frac{M \lambda_h\lambda_g \mathcal{O}( \ln M) }{\lambda_0   \beta^2 \epsilon_0+ M \lambda_h\lambda_g \mathcal{O}( \ln M)}
    \\\label{final ap}  
 \approx &1 -  \frac{1 }{1+\frac{\lambda_0   \beta^2 \epsilon_0}{ M \lambda_h\lambda_g \mathcal{O}( \ln M)}} \rightarrow \frac{\lambda_0   \beta^2 \epsilon_0}{  \lambda_h\lambda_g M \mathcal{O}( \ln M)},
\end{align}
where the last approximation follows from the fact that $\frac{1}{1+x}\approx 1-x$ for $x\rightarrow 0$. 
By increasing $M$, \eqref{final ap} clearly shows that $ {\rm P}\left( E_0\right)  $ approaches zero, and the proof for the lemma is complete.  

     \bibliographystyle{IEEEtran}
\bibliography{IEEEfull,trasfer}
 
   \end{document}